%% file: main.tex
\documentclass[letter,11pt]{article}
\usepackage[utf8]{inputenc}
\usepackage[greek,english]{babel}
\usepackage{fullpage} 
\usepackage{microtype}
\usepackage{mathtools,amssymb,mathrsfs,amsthm}   
\usepackage[table,svgnames]{xcolor}
\usepackage[hyperindex,breaklinks,colorlinks=true,linkcolor=black,citecolor=MidnightBlue,urlcolor=MidnightBlue]{hyperref}
\usepackage{xspace}
\usepackage{enumitem}
\usepackage{tikz}
\usetikzlibrary{matrix,intersections,patterns}

\usepackage[labelfont=bf]{caption}
\usepackage{ifthen}
\usepackage{floatpag}
\usepackage{marvosym}
\usepackage{framed}
\usepackage{url}
\usepackage{breakurl}
\usepackage{algorithm}
\usepackage{algorithmicx}
\usepackage{algpseudocode}
\usepackage[breakable, theorems, skins]{tcolorbox}
\tcbset{enhanced}
\usepackage{stmaryrd}

\newtheorem{theorem}{Theorem}[section]
\newtheorem{definition}[theorem]{Definition}
\newtheorem{corollary}[theorem]{Corollary}
\newtheorem{lemma}[theorem]{Lemma}
\newtheorem{claim}[theorem]{Claim}
\AtEndEnvironment{remark}{\null\hfill\HexaSteel}%

\newcommand{\bLength}{0.85\columnwidth}
\newcommand{\myStrut}{\rule[-.4\baselineskip]{0pt}{\baselineskip}}	%
\newcommand{\LongText}[1]{\parbox[t]{\bLength}{#1\myStrut}}

\newcommand{\done}{\texttt{done}\xspace}
\newcommand{\ctive}{\texttt{active}\xspace}  %

\newcommand{\con}[1]{\overset{#1}{\sim}}   %

\newcommand{\cI}{\mathcal{I}}
\newcommand{\cO}{\mathcal{O}}

\newcommand{\cP}{\mathcal{P}}
\newcommand{\cR}{\mathcal{R}}
\newcommand{\cL}{\mathcal{L}}
\newcommand{\cK}{\mathcal{K}}

\newcommand{\cA}{\mathcal{A}}
\newcommand{\cS}{\mathcal{S}}

\newlength{\hatchspread}%
\newlength{\hatchthickness}%
\newlength{\hatchshift}%
\newcommand{\hatchcolor}{}%
\tikzset{hatchspread/.code={\setlength{\hatchspread}{#1}},
         hatchthickness/.code={\setlength{\hatchthickness}{#1}},
         hatchshift/.code={\setlength{\hatchshift}{#1}},%
         hatchcolor/.code={\renewcommand{\hatchcolor}{#1}}}%
\tikzset{hatchspread=0.9pt,
         hatchthickness=0.3pt,
         hatchshift=0pt,%
         hatchcolor=black}%
\pgfdeclarepatternformonly[\hatchspread,\hatchthickness,\hatchshift,\hatchcolor]%
   {custom north west lines}%
   {\pgfqpoint{\dimexpr-2\hatchthickness}{\dimexpr-2\hatchthickness}}%
   {\pgfqpoint{\dimexpr\hatchspread+2\hatchthickness}{\dimexpr\hatchspread+2\hatchthickness}}%
   {\pgfqpoint{\dimexpr\hatchspread}{\dimexpr\hatchspread}}%
   {%
    \pgfsetlinewidth{\hatchthickness}
    \pgfpathmoveto{\pgfqpoint{0pt}{\dimexpr\hatchspread+\hatchshift}}
    \pgfpathlineto{\pgfqpoint{\dimexpr\hatchspread+0.15pt+\hatchshift}{-0.15pt}}
    \ifdim \hatchshift > 0pt
      \pgfpathmoveto{\pgfqpoint{0pt}{\hatchshift}}
      \pgfpathlineto{\pgfqpoint{\dimexpr0.15pt+\hatchshift}{-0.15pt}}
    \fi
    \pgfsetstrokecolor{\hatchcolor}
    \pgfusepath{stroke}
   }%

\begin{document}

\title{\vspace{-1.5cm}\parbox{1.2\textwidth}{\hspace{-1em}The Topology of Randomized Symmetry-Breaking Distributed Computing}}

\date{} 	%

 \author{
 Pierre Fraigniaud\thanks{Universit\'e de Paris and CNRS, \texttt{pierre.fraigniaud@irif.fr}}
 \and 
 Ran Gelles\thanks{Bar-Ilan University, \texttt{ran.gelles@biu.ac.il}} 
 \and 
 Zvi Lotker\thanks{Bar-Ilan University, \texttt{zvi.lotker@biu.ac.il}}
 }

\maketitle

\begin{abstract}
Studying distributed computing through the lens of algebraic topology has been the source of many significant breakthroughs during the last two decades, especially in the design of lower bounds or impossibility results for \emph{deterministic} algorithms. 
In a nutshell, this approach consists of capturing all the possible states of a distributed system at a certain time as a simplicial complex~$\cP(t)$, called \emph{protocol complex},  and viewing computation as a \emph{simplicial} map from that complex to the so-called \emph{output complex}, $\cO$, that captures all possible legal output states of the system. The topological properties (e.g., homotopy) of the protocol complex depends on the properties of the \emph{input complex}~$\cI$ (i.e., the complex $\cP(0)$ of all possible initial states), and on the communication model. Studying these properties lead to proving or disproving the existence of a simplicial map from $\cP(t)$ to~$\cO$. 

This paper aims at studying \emph{randomized} synchronous distributed computing through the lens of algebraic topology. We do so by studying the wide class of (input-free) \emph{symmetry-breaking} tasks, e.g., leader election, in synchronous fault-free anonymous systems, that is, when the absence of IDs prevents the processing nodes from breaking symmetry in a deterministic manner. 
We show that it is possible to redefine solvability of a task ``locally'', i.e., for each simplex of the protocol complex \emph{individually}, without requiring any global consistency. 
However, this approach has a drawback: it eliminates the topological aspect of the computation, since a single facet has a trivial topological structure. To overcome this issue, we introduce a ``projection''~$\pi$ of both protocol and output complexes, where every \emph{simplex}~$\sigma$ is mapped to a \emph{complex}~$\pi(\sigma)$; the later has a rich structure that replaces the structure we lost by considering one single facet at a time.

We show that a facet of~$\cP(t)$ solves the task if and only if there exists a simplicial map from the projection of that facet to the projection of some facet of the output complex. This approach allows us  to use topological arguments to establish which facets of the protocol complex solve the task, and which do not. 
By doing so, one can compute the probability to solve a task at time~$t$ by analyzing each facet independently, and summing the probabilities of all facets that solve the task. An application of Kolmogorov's zero--one law establishes which symmetry-breaking tasks are \emph{eventually} solvable when $t\to\infty$.

To show the significance and applicability of our approach, we derive necessary and sufficient conditions for solving leader election in synchronous fault-free anonymous  shared-memory and message-passing models. In both models, we consider scenarios in which there might be correlations between the random values provided to the nodes. In particular, different parties might have access to the same randomness source so their randomness is not independent but equal. 
Interestingly, we find that solvability of leader election relates to the number of parties that possess correlated randomness, either directly or via their greatest common divisor, depending on the specific communication model.
\vfil
\end{abstract}

\section{Introduction}
There are two main categories of distributed algorithms: deterministic and randomized. The difference between them is stark as some tasks are solvable  when randomness is present but cannot be solved deterministically. 
Further, randomness often yields faster and more efficient algorithms for tasks solvable deterministically. 
A main example is the task of electing a leader in anonymous systems~\cite{le1977distributed,GHS83}, where $n>1$ identical computing nodes need to designate a single node as their leader. For any deterministic algorithm running in symmetric systems where all nodes are identical, all nodes return exactly the same output. This is known as the impossibility of deterministic algorithms to break symmetry~\cite{angulin80}. On the other hand, if each node has access to an independent source of randomness, then symmetry can be broken, allowing electing a leader~\cite{ASW88}.

It is widely agreed that the major breakthrough of computer science is the invention of the Turing machine~\cite{turing1936}, which abstracts all possible realistic computing machines into a single \emph{universal} model. 
With this abstraction, researchers can ignore the system's specific properties and focus on the computational aspect of the problem rather than on its manifestation in a specific environment.
In the distributed world, while each node is still a Turing machine, the environment plays a large role in setting the relevant model: nodes and communication links may be reliable or subject to various forms of failures. The system may be synchronous or experience arbitrary large communication delays. Also, all nodes might be connected as a single hop or connected in an arbitrary manner. 
This large variety of models required focusing on each model and analyzing each task and algorithm separately. 
In the early 1990s, a unifying framework was developed that captures all the above distributed models: distributed computing via \emph{algebraic topology}~\cite{BG93,HS93,SF93}. 
Similar to the Turing model, 
this framework enables analyzing the complexity of distributed tasks under a unique umbrella  in which the same tools and methodologies apply, independently from the details of the actual distributed model of distributed computing.
However, these pioneering works and the numerous subsequent work that followed them
(see~\cite{Book-topology}, and the references therein) mostly
treated deterministic algorithms, and randomized algorithms remain so far excluded.

The usual topology-analysis of deterministic distributed algorithms~\cite{Book-topology} involves a simplicial complex~$\cP(t)$, called  \emph{protocol} complex, that captures all the possible global states of the distributed system at time~$t\geq 0$, 
where each global state is modeled as a simplex of that complex.\footnote{See Appendix~\ref{app:topology101} for a summary of the topological notions used in this paper.}
Initially $\cP(0)=\cI$, the so-called input complex that captures all initial global states of the system. 
The evolution of the system with time translates to the evolution of the complex~$\cP(t)$, $t\geq 0$. An algorithm solving a task in time~$t$ is then modeled as a simplicial map $\delta:\cP(t)\to\cO$ from the protocol complex to the output complex~$\cO$, where the latter  captures all legal final global states of the system. This mapping may need to satisfy some additional requirements in the case of input-output tasks (e.g., consensus) where the legality of the output depends on the input. 
It turns out that, for many distributed models, certain properties of a topological space, such as connectivity or homotopy type,
are preserved during the evolution of the protocol complex with time. 
Comparing the topological properties of the protocol complex with those of the output complex is a fruitful approach for establishing impossibility results, that is, proving that the mapping~$\delta$ cannot exist. A typical example is the task of consensus in the crash-prone asynchronous shared-memory model in which the protocol complex remains connected throughout time while the output complex is disconnected, preventing the simplicial map~$\delta$ from existing.  

Extending this approach to the analysis of randomized algorithms requires overcoming several difficulties. 
In particular, for a large class of tasks that are not solvable deterministically in anonymous systems, including symmetry breaking tasks such as leader election,  the solvability of these tasks using a randomized algorithm is only \emph{eventual}. 
That is, there is no fixed $t\geq 0$ for which there exists a simplicial map  $\delta:\cP(t)\to\cO$, even if we allow~$t$ to be as large as we wish. 
Nevertheless, the task is eventually solvable, meaning that when $t$~goes to infinity, the probability for solving the task approaches~1. 
One may argue that, for a fixed $t\geq 0$, one can still compute the probability of solving the task at time~$t$ by considering a mapping $\delta_t:V(\cP(t))\to V(\cO)$ (not necessarily simplicial), then considering only the facets of $\sigma\in\cP(t)$ such that $\delta_t(\sigma)\in\cO$, and finally computing the sum of the probability of each of these facets to occur as a function of the outcomes of the random bits used in the protocol. 
This argument is valid. However, it may not be of practical use. Indeed, one is losing the whole point of using topology, as it might be difficult to analyze what topological properties (e.g., homotopy type) are preserved by such an arbitrary map~$\delta_t$. 
Additionally, it might be difficult to relate the topological properties of the ``best'' map $\delta_t$ (i.e., the one that maximizes the probability of solving the task at time~$t$) with the properties of the best map $\delta_{t+1}$ at time~$t+1$. 

In this paper, we give the first stab at developing a topological framework that captures synchronous randomized algorithms of certain types, which we now describe.

\subsection{A Topological Approach for Randomized Algorithms}

We consider input-free symmetry-breaking tasks:  
tasks solely defined by their output complex~$\cO$, which we require to be symmetric (i.e., stable to permutation of the processing nodes names). 
For instance, this is the case of leader election in which there are no constraints on which node may be the leader. 
In this context, a crucial observation is that the solvability of a task at a time~$t$ can be analyzed for each facet $\sigma\in\cP(t)$ separately. 
For each facet $\sigma\in\cP(t)$, we can say that $\sigma$ solves the task~$\cO$ whenever there exists a simplicial map $\delta:\sigma\to\cO$ which maps the state of each node to its output, such that the output value is independent of the actual name of the node. 
The interest of that observation is that it decouples the map $\delta:\sigma\to\cO$  from the map $\delta':\sigma'\to\cO$ that may exist for another facet~$\sigma'$ of~$\cP(t)$. 
The significant drawback of this approach is that, as the aforementioned general approach consisting of considering arbitrary maps $\delta_t:V(\cP(t))\to V(\cO)$, one is losing connection with topology. 
Indeed, the facet $\sigma$ has a trivial topological structure, as it is just a simplex of~$\cP(t)$. 
Our approach, on the other hand, provides a structure to each simplex of~$\cP(t)$; this allows us to keep analyzing the system using topological tools.

To provide the simplices of~$\cP(t)$ with a structure, we ``project'' each simplex~$\sigma\in\cP(t)$ to a sub-complex $\pi(\sigma)$ of $\sigma$, the latter being viewed as a complex.\footnote{To be precise, this complex is the induced subcomplex of~$\cP(t)$ on $V(\sigma)$.} 
Roughly, a set of nodes forms a simplex in $\pi(\sigma)$ if they have an identical individual state in~$\sigma$, where the individual state of a node results from the outcomes of its source of randomness and from its knowledge about the outcomes of the other nodes (acquired during the communications up to time~$t$). 
In a sense, a simplex in $\pi(\sigma)$ bears the meaning that the algorithm up to time~$t$ did not  break symmetry for the nodes associated with this simplex vertices.  
The sub-complex $\pi(\sigma)$ captures the internal state of each one of the parties in a specific execution and portrays similarity in the knowledge of the parties.  
The projection $\pi$ is called \emph{consistency projection}. 
Thanks to the consistency projection, we have regained structure, which we can now utilize to determine the ability of a facet $\sigma\in\cP(t)$ to solve a task~$\cO$. A facet $\sigma$ solves the task~$\cO$ whenever there exists a (name-preserving) simplicial map $\delta:\pi(\sigma)\to\pi(\tau)$ for some facet $\tau\in\cO$, where $\pi(\tau)$ is the projection of $\tau$, that is, the 
sub-complex of $\tau$ in which a set of vertices of $\tau$ forms a simplex in $\pi(\tau)$ if they have identical individual output value in~$\tau$.

Granted with the notion of solvability for each facet of the protocol complex at time~$t$, we show how to define  the \emph{eventual} solvability of the task.  Specifically, let us define the probability $p(t)$ that $\cP(t)$ solves $\cO$ as the sum, taken over all facets of $\cP(t)$ that solve~$\cO$, of the probability of each of these facets. 
Eventual solvability is then based on Borel-Cantelli's Lemma and on Kolmogorov's zero--one law stating that any \emph{tail event} occurs with probability zero or one. 
As a consequence, we can show that $\lim_{t\to\infty}p(t)\in\{0,1\}$, and to forbid mixed answers where tasks are ``solvable'' with some probability~$p\in(0,1)$. Therefore, this provides us with a complete (deterministic) characterization of eventual solvability, namely, an input-free symmetry-breaking task $\cO$ is eventually solvable if $\lim_{t\to\infty}p(t)=1$.

\medbreak 

To sum up this part, given an input-free symmetry-breaking task, i.e., a task represented by a symmetric output complex~$\cO$, we say that a global state~$\sigma \in \cP(t)$ solves~$\cO$ if there exists a name-preserving simplicial map $\delta: \pi(\sigma)\to \pi(\tau)$ for some $\tau\in\cO$, where $\pi$ is the consistency projection. 
For any $t\geq 1$, let $\cS(t)$ be the set of all the global states~$\sigma \in \cP(t)$ that solves the task~$\cO$, and let us define $\Pr[\mathcal{S}(t)] = \sum_{\sigma\in \cS(t)}\Pr[\sigma]$.
We say that $\cO$ is eventually solvable if and only if $\lim_{t\to \infty} \Pr[\mathcal{S}(t)]=1$.

\subsection{Application to Leader Election in Anonymous Models}

To demonstrate the interest and usability of the topological approach of randomized algorithms sketched in the previous section, we study the arguably most prominent symmetry-breaking task, namely \emph{leader election}, in two important models for anonymous computing, that is, the blackboard model and the message-passing model with port-numbers. 
While in this paper we focus on leader election, we note that the vast majority of other symmetry-breaking tasks in anonymous message-passing networks are trivially solvable once a leader can be elected; see~\cite{EPSW14}, for instance.\footnote{Another reason to focus on leader election is that many other important tasks, such as reaching consensus or $k$-set agreement, are \emph{deterministically}  solvable  in the fault-free setting. See also Appendix~\ref{app:ni-task}.} %
Moreover, we consider an exhaustive range of randomness source assignments to the nodes, from private randomness (each node has its own independent source of randomness) to shared randomness (all nodes get their random bits from the same source). Indeed, in many real-life situations, we observe dependencies and correlations between the different randomness sources, and this is in particular noticeable in distributed systems as the parties are all deterministic machines that obtain their randomness via pseudo-random generators, which might cause the randomness of some (or all) parties to be correlated or even identical. 
This is not just a theoretical whim: recent work showed that the same random SSH keys were generated by more than 250,000 ``independent'' devices~\cite{ssh15} and that 1 out of 172 RSA-based certificates found online (i.e., about 450,000 certificates) have a random RSA key that shares a factor with the key of another ``independent'' certificate, allowing the complete factorization of these RSA keys and breaking the security of these certificates~\cite{KV19}. 

Specifically, we model our setting as follows.
The system is composed of~$n$ parties, where all parties are connected to each other either via an anonymous broadcast channel (blackboard model) or via private point-to-point channels (message-passing model). In the latter model, each private channel of a party is locally identified by an index in $\{1,\dots,n-1\}$, called port-number~\cite{JS85}. Other than that, the $n$ parties are identical and anonymous (they have no IDs). 
There are $k$, $1\leq k\le n$, independent sources of randomness~$\mathbf{R}_1,\ldots, \mathbf{R}_k$, where each randomness source generates an infinite sequence of i.i.d random bits, each bit being picked uniformly at random. Each party is connected to a single randomness source. However, different parties may be connected to the same source, hence, their randomness is completely dependent (i.e., identical).
This modeling is simple, but it allows us to capture both the case where all parties have independent randomness as well as the case where all parties share the same randomness source, or any situation in between---where some parties are correlated and others are not. 

One can easily verify that the correlations between parties bring up
interesting questions. In particular, under which assumption on the correlations between the parties is 
leader election possible? There are trivial answers under specific scenarios. For instance, if each party is connected to a private source of randomness, then leader election is eventually solvable, even in the weak blackboard model (e.g., whenever one party gets a bit~1 while all the other parties get bits~0). Instead, if all parties are connected to the same randomness source, then symmetry cannot be broken in the blackboard model, and leader election is impossible.  But what about intermediate cases? And what about the message-passing model in which the  port-numbers may assist in breaking symmetry? 

\medskip

We show that, in the blackboard model, leader election is eventually solvable if and only if there exists a randomness source~$\mathbf{R}_i$, $i\in\{1,\dots,k\}$, that is connected to a single party (Theorem~\ref{thm:LE-BB}). 
In the message-passing model, we show that whenever $n_i\geq 1$ parties are connected to the randomness source~$\mathbf{R}_i$,  $i\in\{1,\dots,k\}$, leader election is eventually solvable if and only if $\gcd(n_1,n_2,\ldots,n_k)=1$ (Theorem~\ref{thm:le-PN}). 
The latter result must be read as follows. 
If $\gcd(n_1,n_2,\ldots,n_k)=1$, then leader election is eventually solvable for every assignment of the port-numbers to the channels, and if $\gcd(n_1,n_2,\ldots,n_k)>1$, then there is an assignment of the port-numbers to the channels such that leader election is not eventually solvable. 
The intuition behind these two theorems is somewhat immediate.
In the blackboard model, if two processes (or more) possess the same randomness source, then their states  will be identical throughout the computation, hence none of them can be elected a leader. On the other hand, if there is a single process with a unique source of randomness, then with probability 1 this process eventually arrives at a state distinct from all other processes. In the next round, every process will see that this has happened and elect the distinct process as a leader. 
In the message-passing model, the GCD condition guarantees that processes with the same randomness source $\mathbf{R}_i$ can ``match'' themselves with processes that have a different randomness source~$\mathbf{R}_j$. At the end of this step, any process in the smaller set will have a match from the other set. 
The key observation is that if we deactivate all the matched-processes that belong to the \emph{larger} set, then the GCD condition still holds. Thus, we can repeat this matching process recursively, similar to Euclid's GCD algorithm~\cite{Euclid}, until we reach a set with a size of~1. That process will be the leader.
One can verify that this idea cannot work when the GCD is not~1, since eventually one of the sets will have a size of~0, while the size of the other set will be a multiple of the GCD. 
Indeed, if the GCD is~$g$, we show that every process belongs to a set (of size $c\cdot g$ for some $c\in \mathbb{N}$) such that all processes within the same set share the exact same state.

While the above intuition could lead to a direct proof for leader election solvability, it is almost straightforward to derive  some of the above insights by applying our topological framework. 
The above intuition might seem trivial in hindsight, but such insights might not be simple to come up with a priori. As an example, consider the conditions for the solvability of 2-leader election, i.e., obtaining exactly two leaders. The topological framework immediately leads to the correct characterization of this problem and many others. We encourage the reader to find a direct characterization in both the blackboard and message-passing models, and then compare it with the characterization obtained via the topological framework. 

\medskip
In order to establish the above characterization, we apply the topological approach described in the previous sub-section. Leader election is represented by the output complex $\cO_{\mathsf{LE}}$ defined as follows. 
A set $\{(i,x_i):i\in \{1,\dots,n\}\}$ is a facet of $\cO_{\mathsf{LE}}$ if $\{x_1,\dots,x_n\}=\{0,1\}$, and there is a unique $i\in\{1,\dots,n\}$ such that $x_i=1$. 
See Figure~\ref{fig:OLE} for an illustration of $\cO_{\mathsf{LE}}$ and its consistency projection.
First, we characterize the global states~$\sigma \in \cP(t)$ that solve~$\cO_{\mathsf{LE}}$, that is, the states for which there exists a name-preserving simplicial map $\delta: \pi(\sigma)\to \pi(\tau)$ for some $\tau\in\cO_{\mathsf{LE}}$, where $\pi$ is the consistency projection.
Then, given the set $\cS(t)$  of the global states~$\sigma \in \cP(t)$ that solves leader election, we compute $\Pr[\mathcal{S}(t)\mid \alpha] = \sum_{\sigma\in \cS(t)}\Pr[\sigma\mid \alpha]$ for every configuration $\alpha$ of the connections between the parties and the randomness sources. 
Finally, we compute $\lim_{t\to \infty} \Pr[\mathcal{S}(t)\mid\alpha]$ for figuring out under which condition this limit is~1, which establishes eventual solvability.

Computing the probability of a global state $\sigma\in\cP(t)$ to occur (i.e., computing $\Pr[\sigma\mid \alpha]$)  is, however, non-trivial. 
So, we introduce another complex, called the \emph{realization} complex at time~$t$, denoted by $\cR(t)$. The vertices of $\cR(t)$ are pairs $(i,x_i)$ where $i\in\{1,\dots,n\}$, and $x_i\in\{0,1\}^t$. Each set of vertices $\{(i,x_i):i\in I\}$ is a simplex of $\cR(t)$, for every non-empty set $I\subseteq\{1,\dots,n\}$. 
The interest of $\cR(t)$ is that given an assignment $\alpha$ of the randomness sources to the nodes and given a facet $\rho$ of $\cR(t)$, it is easy to compute $\Pr[\rho \mid \alpha]$. Interestingly, it turns out that the facets of~$\cR(t)$ are isomorphic to the facets of~$\cP(t)$. That is, we show that a specific realization at time~$t$ of the randomness sources fully and uniquely defines the states of the parties at time~$t$ in the protocol complex~$\cP(t)$. This makes the computation of $\Pr[\sigma\mid \alpha]$ easier for each facet $\sigma$ of $\cP(t)$, by computing the probability $\Pr[\rho\mid \alpha]$ of the corresponding facet $\rho$ of $\cR(t)$. 

Thanks to the isomorphism between the facets of~$\cP(t)$ and~$\cR(t)$, we can focus on determining whether a given facet $\rho$ of $\cR(t)$ solves leader election. For this purpose, we introduce a variant of the consistency projection, denoted by $\tilde{\pi}$, that applies to facets of~$\cR(t)$. Then we study the existence of a name-preserving simplicial map $\delta:\tilde{\pi}(\rho)\to\pi(\tau)$ for some facet $\tau\in\cO_{\mathsf{LE}}$ for determining whether the facet $\rho$ of $\cR(t)$ solves leader election or not. 
Since for any facet~$\tau\in\cO_{\mathsf{LE}}$, $\pi(\tau)$ contains an isolated vertex, for any realization~$\rho$ that potentially solves leader election, $\tilde\pi(\rho)$ must contain an isolated vertex as well. 
In the blackboard model, the dimension of the smallest facet in~$\tilde\pi(\rho)$ is $\min\{n_1,n_2,\ldots,n_k\}-1$, and thus the presence of an isolated node in~$\tilde\pi(\rho)$ requires that  $\min\{n_1,n_2,\ldots,n_k\}=1$.
In the message-passing model, we prove that there exists some way to assign port-numbers to the channels such that, for any facet $\gamma\in\tilde\pi(\rho)$ with dimension~$d$, it holds that
$\gcd(n_1,n_2,\ldots,n_k) \mid d + 1$. Therefore, if $\gcd(n_1,n_2,\ldots,n_k)>1$, then there are no isolated nodes in~$\tilde\pi(\rho)$, and the task cannot be solved. 
In order to prove the other direction, i.e., that if 
$\gcd(n_1,n_2,\ldots,n_k)=1$ then leader election is eventually solvable, we describe an algorithm that imitates Euclid's algorithm~\cite{Euclid} for computing the greatest common divisor of the dimensions of the facets in~$\tilde\pi(\rho)$, until reaching an isolated vertex.

\subsection{Related Work}
As mentioned above, the pioneering work~\cite{BG93,HS93,SF93} formulated distributed computations in the language of algebraic topology in order to show impossibility results in the presence of failures.
A tremendous amount of subsequent work is described in the book of Herlihy, Kozlov, and Rajsbaum~\cite{Book-topology} (see also dozens of citations therein). We mention that the language of algebraic topology was found useful to analyze systems both in the \emph{message-passing} model as well in the \emph{shared-memory} model, and it can even be extended to capture non-benign faults, like Byzantine failures~\cite{MTH14}.

\smallskip
Leader election has been extensively studied as an interesting special case of symmetry-breaking between nodes, usually in anonymous systems.
Angluin~\cite{angulin80} showed that no deterministic algorithm could elect a leader in anonymous networks (usually, in a ring) while it is possible to elect a leader non-deterministically or probabilistically.
A long line of leader election algorithms, as well as lower bounds, were developed for certain special cases%
~\cite{AAGHK89,AM94,ASW88,DMR08,IR90,KPPRT15,peleg90,SS94}. 
Mostly related to this paper is the work of
Yamashita and Kameda~\cite{YK96}, which fully characterize the solvability of deterministic leader election over general graphs in the message-passing model; their characterization follows from considering various types of \emph{symmetries} in the graph. 
Boldi et al.~\cite{BSVXGS96}  also give a full characterization of solvability for leader election in networks with and without port-numbers, using a method of graph-homomorphisms known as graph covering or fibrations. This was later extended by Chalopin et al.~\cite{CGM12} to families of graphs. (Codenotti et al.~\cite{CGPS97} mention that leader election in $K_{m,n}$ is possible if and only if $\gcd(m,n)=1$).
In contrast to~\cite{BSVXGS96,CGM12,YK96}, which require a complicated analysis of the structure of the network, 
our characterization is much more straightforward and intuitive, and is based only on the greatest common divisor of the sizes of the subsets of parties connected to the same randomness source --- however, this clean characterization applies only for the clique.
There is plenty of other work on leader election in various models and settings. We surveyed above only the most relevant work to our model.  
No prior work considered the interesting case of symmetry-breaking in correlated randomness settings to the best of our knowledge.
We stress again that the analysis of leader election is merely a single example of our framework and machinery for topological analysis of randomized distributed algorithms.

\section{The Model}

\subsection{Communication and Randomness}

\paragraph{Communication model.}
We consider $n\geq 1$ identical fault-free processing nodes with no identifiers (i.e., they are anonymous) running the same algorithm. The nodes perform computation and communication in lockstep, that is, we assume synchronous rounds. For $r\geq 1$, the $r$-th round occurs between time~$r-1$ and time~$r$. During each round every node can send a message to each other node, and can receive messages from the other nodes. The size of the messages is finite but unbounded. Without loss of generality, we can assume that each node sends its entire history to the other nodes at every round, which is a complete description of all the information accumulated by the node during the previous rounds, including its inputs, the content of the messages that were sent and received, when these messages were sent and received, etc. We consider two sub-models regarding the way each node communicates with the other nodes in the system.
\begin{itemize}
\item
\emph{The blackboard model:}
There is a shared memory called \emph{blackboard}, and each node can send information to the other nodes by appending a message to the blackboard. Every message written on the blackboard by a node at the beginning of a round can be seen by all the other nodes at the end of the round. However, there are no indications about which node is the origin of a message written on the board. Furthermore, the order in which the messages appear on the blackboard during a single round is arbitrary; without loss of generality we will assume all the messages written to the board in a single round appear on it in a lexicographic order.
\item \emph{The message-passing model:}
The nodes are connected as a clique~$K_n$, and they communicate by passing messages through the edges of~$K_n$. A message sent by a node~$u$ through its incident edge~$e$ at the beginning of a round is received by the other extremity of~$e$ at the end of the round. The $n-1$ edges incident to every node~$u$ are labeled by $n-1$ distinct integers in $\{1,\ldots, n-1\}$. The label given to an edge~$e$ incident to node~$u$ is called the port number of~$e$ at~$u$. The port numbers are arbitrary, and there are no correlations between the two port numbers of an edge at its two extremities. Our results hold for the worst case assignment of port numbers, that is, they can be assumed to be assigned by an adversary. 
\end{itemize}
In both cases, there are no restrictions on the amount of information that can be transmitted during a round, that is, there are no restrictions on the size of the messages to be written on the blackboard or to be sent through the links of the network. In other words, we assume \emph{full information} protocols. 

\paragraph{Randomness.}

Given a positive natural number~$m$ we denote by~$[m]$ the set~$\{1,\ldots,m\}$. 
We assume that the system has access to $k$ independent  sources of randomness, denoted by $\mathbf{R}_1,\ldots,\mathbf{R}_k$, for $k\in[n]$. During every round, each source $\mathbf{R}_i$, $i\in [k]$, generates a single bit whose value is chosen uniformly at random in $\{0,1\}$. The random variable equal to the infinite binary string generated by $\mathbf{R}_i$ is denoted by~$R_i$.
Each node is connected to one of the $k$ sources $\mathbf{R}_i$, $i\in [k]$, and it may be the case that several nodes are connected to the same source of randomness (this necessarily  happens whenever $k<n$). The random variable equal to the infinite binary string received by node~$i\in[n]$, is denoted by~$X_i$. 
At time~$t$, node~$i$ has received a prefix of length~$t$ of $X_i$, i.e., a $t$-bit string $x_i(t)\in\{0,1\}^t$.

Throughout the paper, random variables are denoted by uppercase letters (e.g., $R_i, X_i$), and their realizations are denoted by lowercase letters (e.g., $r_i, x_i$). A random variable $Z$ at round~$t$ is denoted by $Z(t)$, and for a string~$S$, we let $S(t,\dots,t')$ denote the sub-string $S(t)\, S(t+1) \dots S(t')$; the same holds for realizations of random variables, etc. 

\subsection{Knowledge}

For $t\geq 0$, let $K_i(t)$ be the \emph{knowledge} of node~$i\in[n]$ at time~$t$, defined recursively as follows. For every $i\in [n]$, $K_i(0) = v_i$, where $v_i$ is the input value given to node~$i$. If there are no inputs, then $K_i(0) = \bot$. In the blackboard model, for $t\geq 1$, we set
\begin{equation}\label{eqn:knowledgeBB}
K_i(t) = \Big (K_i(t-1), X_i(t), \big\{K_j(t-1) : j\in [n]\smallsetminus\{i\}\big\} \Big).
\end{equation}
where the knowledge $\{K_j(t-1) : j\in [n]\smallsetminus\{i\}\}$ received from the other nodes is a multi-set; this multiset corresponds to the entire content of the blackboard, up to the order which is lexicographic by assumption.
In the message-passing model, we set
\begin{equation}\label{eqn:knowledgeMP}
K_i(t) = \Big(K_i(t-1), X_i(t), \big ( K_{\pi_i(1)}(t-1),\dots, K_{\pi_i(n-1)}(t-1) \big) \Big),
\end{equation}
where $\pi_i(j)\in[n]$ denotes the node connected to node~$i$ by the edge with port-number~$j$ at~$i$. Note that node~$i$ does not know~$i$, nor does it know~$\pi_i(j)$ for $j=1,\dots,n-1$. Note also that, at time~$t$, a node knows the $t$ random bits it received from its source of randomness during the first $t$ rounds, but only the $t-1$ bits received by every other node from its source of randomness during the first $t-1$ rounds.

\section{Topological Description of Randomized Symmetry-Breaking Distributed Algorithms}
\label{sec:topology}

In this section, we describe the topological framework that enables the analysis of distributed algorithms, and extends it to capture the analysis of randomized algorithms. In Section~\ref{sec:le} we will later show how to actually use this framework for analyzing the solvability of leader election as a function of the randomness given to each node, for both blackboard and message-passing models. The reader is referred to Appendix~\ref{app:topology101} for some basic topological definitions. Further information can be found, e.g., in~\cite{Book-topology}.

\subsection{Topological Setting} 

We recall the notion of tasks, and of solvability of tasks in fixed time, within the topological framework. 

\paragraph{Tasks.}

A \emph{task} is described by a triple $\Pi=(\cI,\cO,\Delta)$, where $\cI$ is the input complex, $\cO$ is the output complex, and $\Delta:\cI\to 2^{\cO}$ is the input-output specification of the task (see, e.g., \cite{Book-topology}). All the complexes are ``colored'', in the sense that their vertices have the form $(i,x)$ with color $i\in[n]$, for some value~$x$, and their simplices include vertices with different colors. We rather refer to the color~$i$ of a vertex $(i,x)$ as its \emph{name},  i.e., the node named~$i$ holds value~$x$; we denote $\mathsf{name}((i,x))=i$, which can be extended to a set of nodes in a straightforward manner.

In this paper, we focus on input-free symmetry breaking tasks, so $\cI$ is the trivial complex with a single facet $\{(i,\bot):i\in [n]\}$. For input-free tasks, the input-output specification is trivial, that is, given any input simplex $\sigma\in\cI$, this simplex is mapped to all simplices of~$\cO$ with same set of names as~$\sigma$.  
A symmetry-breaking task is thus simply defined by its output complex~$\cO$.
We only require that the output complex must be stable by permutation of the colors of the processes. That is, if $\{(i,v_i):i\in I\}$ is a simplex of~$\cO$, with $\varnothing\neq I\subseteq [n]$, then, for every permutation $\pi:I\to I$, it must be the case that $\{(i,v_{\pi(i)}):i\in I\}\in \cO$. 
For instance, the output complex $\cO_{\mathsf{LE}}$ for leader election (\textsf{LE}) has $n$ facets 
\[
\tau_i=\{(1,0),\dots,(i-1,0),(i,1),(i+1,0),\dots,(n,0)\},
\]
for $i=1,\dots,n$. That is, $\tau_i$ is the legal output state in which node~$i$ is elected, and the $n-1$ other nodes are defeated. 
Note that $\cO_{\mathsf{LE}}$ is symmetric.

\paragraph{Communication and randomness configuration.}

The exchanges of information between the nodes occurring throughout the execution is captured by the protocol complexes. The vertices of the protocol complex at time~$t$, denoted by $\cP(t)$, are pairs $(i,K_i(t))$, $i\in [n]$, where $K_i(t)$ denotes the knowledge acquired by node~$i$ at time~$t$. 
See Figure~\ref{fig:Pwith2parties} for a demonstration of $\cP(t)$ for $t=0,1,2$ for a computation with two parties.
\begin{figure}[htp]
    \centering
    \input{tikz_P}
    \caption{\small The evolution of a 2-party algorithm for time steps $t=0,1,2$. 
    The knowledge of each party at a given time, written next to the respective node, consists of the party's previous knowledge, the random bit it has achieved in that round, and the knowledge of the other party sent to it in the previous round. In this figure, $k_0 = (\bot,0,(\bot))$, $k_1 = (\bot,1,(\bot))$. Each edge is a possible state of the system, whose probability is determined by the specific randomness configuration~$\alpha \in \cA$ in a given execution. An edge at time~$t$ (i.e., a facet of~$\sigma\in \cP(t)$)  evolves into $4$ possible facets (edges) of~$\cP(t+1)$. These correspond to the 4 possible values of the random bits obtained by the two parties at time~$t+1$.}
    \label{fig:Pwith2parties}
\end{figure}
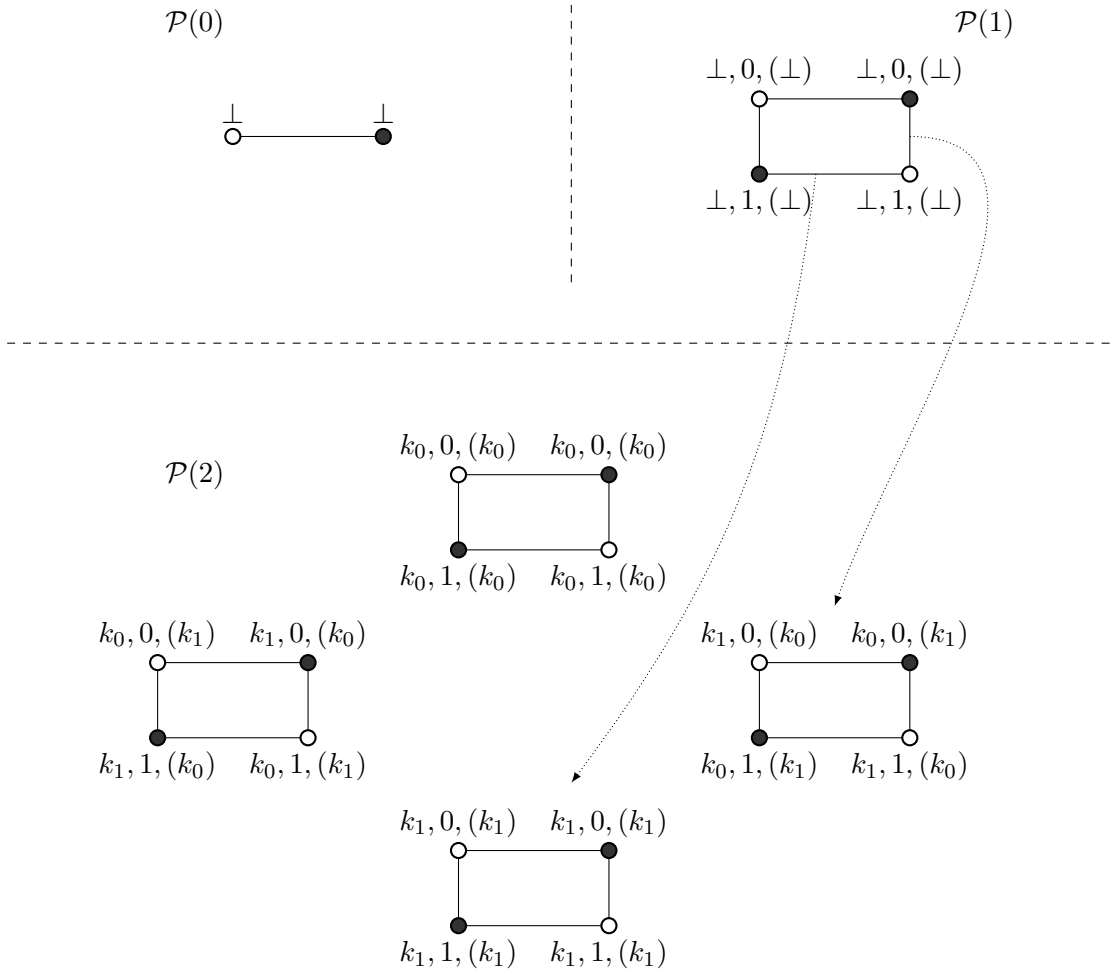
The knowledge acquired by the nodes however depends the randomness they obtain and on the way the $k$~randomness sources are assigned to the $n$~nodes. We also use a complex to formalize all possible assignments. This complex, denoted by~$\cA$ (for ``assignment'') is the pure $(n-1)$-dimensional complex whose facets are of the form 
\[
\alpha=\{(1,j_1), \ldots, (n,j_n)\},
\]
where 
$\bigcup_{i=1}^n \{j_i\} =[k]$ for some $k\in [n]$; that is, without loss of generality we rename the $k$ different sources to be contiguous in $\{1,\ldots, k\}$.
For every $i\in [n]$, a pair $(i,j)$ means that node~$i$ is connected to the randomness source~$\mathbf{R}_{j}$. Every facet~$\alpha$ of~$\cA$ is called a randomness-\emph{configuration}, that is, a configuration determines which node is connected to which randomness source, for all the nodes. 
For a given configuration~$\alpha$ we denote by $k=k(\alpha)$ the number of different sources actually connected to the systems. Note that our restriction on~$\cA$'s facets means that the $k$ sources actually connected to the system are exactly  $\mathbf{R}_1,\ldots,\mathbf{R}_k$.
A set 
\[
\{(i,K_i(t)):i\in [n]\}
\]
forms a facet of $\cP(t)$ whenever there exists a configuration $\alpha\in\cA$ such that, with non-zero probability, each node~$i$ acquires knowledge $K_i(t)$, $i=1,\dots,n$, after $t$ rounds of communication.

\paragraph{Solvability in fixed time.}

Recall that a map $\delta$ between the vertex sets of two complexes is \emph{simplicial} if it preserves simplices. A simplicial map between two chromatic complexes is \emph{name-preserving} if it preserves the names of the vertices (i.e., for every vertex $(i,x)$, $\delta(i,x)=(i,y)$ for some~$y$ that may depend on~$i$ and~$x$), and it is \emph{name-independent} if it is oblivious to the names (i.e., if $\delta(i,x)=(i,y)$, then $\delta(j,x)=(j,y)$ for every $j$, that is, $y$ depends solely on~$x$).  
In this work all our complexes are chromatic and all the maps are name-preserving.

In the standard topological setting, a task $(\cI,\cO,\Delta)$ is solvable in $t$ rounds if there exists a name-preserving and name-independent simplicial map 
\[
\delta:\cP(t) \to \cO. 
\]
This notion of solvability is not appropriate to our randomized setting, for two reasons. First, we want to discuss solvability as a function of the randomness-configuration $\alpha\in\cA$ of the randomness sources. Second, and more importantly, there might be no~$t$, even arbitrarily large, enabling such a simplicial map $\delta:\cP(t) \to\cO$ to exist. 
This holds even when the task is eventually solvable under the configuration~$\alpha$.
To better illustrate this point, assume two processes, each with its private and independent source of randomness. There is no~$t$ for which one can guarantee that the two processes have received two different bits at some round $r\in\{1,\dots,t\}$. Yet, leader election is almost surely solvable in this context as, eventually, the two processes will receive two different bits at some round. 

\subsection{Eventual Solvability}

A global state of the system at time~$t$ is a facet $\sigma$ of~$\cP(t)$.

\begin{definition}\label{def:solve_P}
A global state~$\sigma \in \cP(t)$ solves a symmetry-breaking task~$\cO$ if there exists $\tau\in\cO$ and a \emph{name-preserving} and \emph{name-independent} simplicial map 
$
\delta: \sigma\to \tau.
$
\end{definition} 

This definition of solvability for a facet of $\cP(t)$ is motivated by the following observation. Let us assume that, for a facet $\sigma=\{(i,K_i(t)):i\in[n]\}$ of~$\cP(t)$, there is a name-preserving and name-independent simplicial map $\delta: \sigma\to \tau$. This map can be written as $\delta(i,K_i(t))=(i,f(K_i(t))$, $i=1,\dots,n$, for some function~$f$ of the knowledge. Since the output complex $\cO$ is symmetric, the map~$\delta$ yields the existence of a  simplicial map $\delta': \sigma'\to \tau'$ for every simplex $\sigma'=\{(i,K_{\pi(i)}(t)):i\in[n]\}$ of~$\cP(t)$ where $\pi:[n]\to [n]$ is a permutation, letting
$\tau'=\{(i,v_{\pi(i)}) : (i,v_i)\in\tau \} \in\cO$ and defining
$\delta'(i,K_{\pi(i)}(t))=(i,f(K_{\pi(i)}(t)))$. 
From this observation, we derive an algorithm solving $\cO$ whenever the global state of the system is of the form $\sigma'=\{(i,K_{\pi(i)}(t)):i\in[n]\}$ for some permutation~$\pi:[n]\to[n]$. Indeed, at time~$t$, the knowledge accumulated by nodes during the first $t$ rounds results in some global state $\sigma'\in \cP(t)$ of the system. Each node may not be aware of $\sigma'$ as its individual knowledge may also be compatible with other global states. Nevertheless, after one more round, the nodes receives the knowledge of the other nodes in $\sigma'$. This enables each node to reconstruct $\sigma'$ up to a permutation of the names of the other nodes. By applying $f$ on its knowledge, every node can then compute its output such that the collection of all outputs truly solves the task at time $t+1$ whenever the processes were in global state~$\sigma'$ at time~$t$. 

Given an assignment $\alpha\in \cA$ of the randomness sources to the nodes, every global state $\sigma \in \cP(t)$ has some probability to occur at a given time. One can thus compute the probability of solving the task at time~$t$ given $\alpha$ as
\[%
\Pr[\mbox{$\cP(t)$ solves $\cO$} \mid \alpha] = \sum_{\mbox{\small $\sigma$ solves $\cO$}} \Pr[\sigma \mid \alpha].
\]%
Observe that whenever $\sigma\in \cP(t)$ solves~$\cO$ via some $\delta_{\sigma}:\sigma\to\tau$, every global state $\sigma'\in\cP(t')$ for $t'\geq t$ that results from $\sigma$ after $t'-t$ additional rounds also solves~$\cO$. This is simply because the knowledge is cumulative, and one can discard all the additional information obtained by the nodes during the $t'-t$ additional rounds for defining $\delta_{\sigma'}:\sigma'\to \tau$ using~$\delta_{\sigma}$. More importantly, the following holds. 

\begin{lemma}
For every input-free symmetry-breaking task $\cO$, and every randomness-configuration~$\alpha\in\cA$, 
\[
\lim_{t\to \infty} \Pr[\mbox{$\cP(t)$ solves $\cO$} \mid \alpha]\in\{0,1\}.
\]
\end{lemma}

\begin{proof}
Let us assume that $\Pr[\cP(t) \; \mbox{solves} \; \cO \mid \alpha]>0$. Therefore, there exists a global state  $\sigma\in \cP(t)$ that solves $\cO$, where $\Pr[\sigma \mid \alpha]>0$. 
This means that there exists a set of nodes' knowledge $K = \{K_i(t)\}_{i\in [n]}$ that yields a solution to the task. 
Note that in an input-free task, knowledge (at time~$t$) stems only from the randomness and messages sent by time $t$. That is, there exists $k$ realizations of randomness, each of length~$t$, that induce the set $ K$, and these realizations have non-zero probability to occur, given~$\alpha$. 
Denote these realizations by~$R$. 
For every $s\geq 0$, let $E_s$ be the event ``$R$ occurred during rounds $s+1$ to $s+t$''. Since knowledge is cumulative, the occurrence of~$E_s$ implies that, at time $s+t$, the nodes hold a knowledge $\{K_i(s+t)\}_{i\in [n]}$ where, for every $i\in[n]$, $K_i(t)$ is included in $K_i(s+t)$. 
Thus, for any $s\geq 0$, if $E_s$ holds, the system reaches a global state that solves the task.
Furthermore, note that two events $E_s$ and $E_{s'}$ are independent whenever $|s-s'|>t$, since our sources are i.i.d across time.

Recall that, given an infinite sequence $(X_i)_{i\geq 1}$ of independent random variables, a \emph{tail event} for $(X_i)_{i\geq 1}$ is an event based on the realization of the $X_i$'s, $i\geq 1$, which is  probabilistically independent of any \emph{finite} subset of $\{X_i:i\geq 1\}$. Kolmogorov's zero--one law~\cite{HandbookOfProb2014Ch2} states that, for any tail event~$E$ over $(X_i)_{i\geq 1}$, either
\begin{align*}
\Pr[E]=0,\quad \text{or}\quad \Pr[E]=1.
\end{align*}
Let 
$
E=\bigcup_{s=0}^\infty E_{s t}.
$
We have 
$
\Pr [E\mid \alpha ] \le \lim_{t\to \infty} \Pr[\mbox{$\cP(t)$ solves $\cO$}\mid \alpha].
$
Moreover, $E$ is a tail event, and therefore its probability is either~0 or~1. Since $\Pr[E_0\mid \alpha]>0$, it follows that  $\Pr[E\mid \alpha]=1$. 
We conclude that 
if $\Pr[\cP(t)$\text{ solves }$\cO \mid \alpha]>0$, then 
$
\lim_{t\to \infty} \Pr[\mbox{$\cP(t)$ solves $\cO\mid \alpha$}]=1,
$
as claimed. 
 \end{proof}

This result motivates the following definition. 

\begin{definition}\label{def:solvabilityAS}
A task $\cO$ is  \emph{eventually solvable} given the randomness-configuration~$\alpha\in\cA$ if and only if
\[
\lim_{t\to \infty} \Pr[\mbox{$\cP(t)$ solves $\cO$} \mid \alpha]=1.
\]
\end{definition}

\subsection{Realization and Consistency Complexes}

We now introduce new complexes, which are essentially reformulations of the  protocol and output complexes~$\cP(t)$ and~$\cO$, more suitable for the analysis of input-free symmetry-breaking tasks.  

\paragraph{Realization complex.} 

In the setting of this paper, namely the anonymous blackboard and message-passing models, the protocol complex is entirely determined by the values of the random bits produced by the $k$ sources of randomness. The realization complex at time~$t$ is denoted by $\cR(t)$, for any $t\geq 1$. The vertices  of $\cR(t)$ are pairs $(i,x_i)$ where $x_i\in\{0,1\}^t$ is a binary string of length~$t$ (the random bits received by node~$i$ during the first~$t$ rounds). Specifically, $\cR(t)$ has vertex-set 
\[
V(\cR(t)) = \{(i,x_i) : i\in [n] ,  x_i \in \{0,1\}^t\}. 
\]
For $I\subseteq [n]$,  a set $\rho = \{(i,x_i) :  i\in I\}\subseteq V(\cR(t))$ is a simplex of $\cR(t)$ if there exists a randomness-configuration~$\alpha\in\cA$ such that, with non-zero probability, each node $i\in I$ may receive the random bit-string $x_i$. 
See Figure~\ref{fig:tikz_R} for an illustration with three parties.
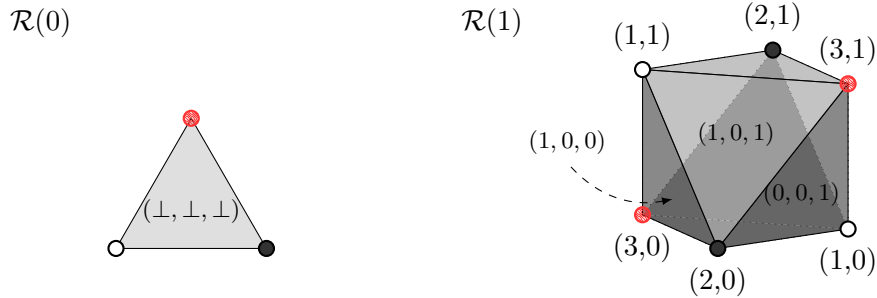
\begin{figure}[htp]
    \centering
    \input{tikz_R}
    \caption{\small A demonstration of $\cR(0)$ and $\cR(1)$ in a system with 3 processes. Each facet represents a possible state of the system described via the randomness received by the parties up to that time. The notation $(w,b,r)$ with $w,b,r\in\{0,1\}$ describes the randomness of the white, black, and red nodes accordingly, i.e., the simplex $\{(1,w), (2,b), (3,r)\}$; $\bot$ is the empty string.}
    \label{fig:tikz_R}
\end{figure}

We observe relations between~$\cP(t)$ and~$\cR(t)$ that allow us to analyze the algorithm via the more intuitive~$\cR(t)$. For any time $t\ge 0$ there exists a simplicial map $h: \cP(t) \to \cR(t)$ that takes each vertex $(i,K_i)\in \cP(t)$ to $(i,x_i)\in \cR(t)$, where $x_i\in\{0,1\}^t$ is the randomness received by party~$i$ according to $K_i(t)$; recall that $K_i(t)$ indeed contains $x_i(t)$ by its definition (Eqs.~\eqref{eqn:knowledgeBB} and~\eqref{eqn:knowledgeMP} for the blackboard and message-passing models, respectively). Note that $K_i(t)$ contains also randomness received by all the other parties up to round $t-1$, hence, $h$ maps multiple vertices to~$(i,x_i)$. Note that $h$ is name-preserving by construction.

\label{sec:isomorphism}
We observe that the simplicial map~$h$ induces an isomorphism between \emph{facets} of $\cP(t)$ and facets of~$\cR(t)$. Indeed, a facet $\{(i,K_i(t)):i\in [n]\} \in \cP(t)$ uniquely determines the randomness $(x_1,\ldots, x_n)$ received by all parties up to round~$t$, and is mapped to the facet $\{ (i,x_i) : i\in [n]\}$ of~$\cR(t)$. Similarly, if one determines the randomness by time~$t$, $(x_1,\ldots, x_n)$, this uniquely defines the knowledge every party holds up to time~$t$, since each $K_i(t)$ consists of $x_i$ and $K_j(t-1)$ for $j\in[n]$, and these, by induction, are deterministic function of $(x_1,\ldots, x_n)$.
With a slight abuse of notation we will commonly refer to~$h$ as an \emph{isomorphism}, implicitly restricting it to act on facets.

\paragraph{Consistency complexes.} 

We now consider two general ``consistency-projections'' $\pi$ and $\tilde{\pi}$ that apply on chromatic complexes. Let $\cK$ be a pure chromatic complex of dimension $n-1$, that is, a complex whose vertices are pairs of the form $(i,v)$, with $i\in[n]$, and $v$ a value. 
Let $\sigma=\{(i,v_i): i\in [n]\}$ be a facet of~$\cK$. We define the complex $\pi(\sigma)$ as the complex on vertex-set $\{(i,v_i) :  i\in [n]\}$ such that, for every non-empty $I\subseteq[n]$, 
\begin{equation}\label{eqn:pi}
\{(i,v_i) :  i\in I\} \in  \pi(\sigma) \iff \forall (i,j)\in I\times I, v_i=v_j.
\end{equation}
The projection $\pi$ applied simultaneously to all the facets of $\cK$ results in the complex
\[
\pi(\cK)=\bigcup_{\sigma\in\cK}\pi(\sigma), 
\]
where the union is taken on the facets of~$\cK$. 
We note that $\pi(\cK)$ is a subcomplex of~$\cK$.

As a simple illustrative example, in the case of leader election (\textsf{LE}), the complex $\pi(\cO_{\mathsf{LE}})$ has facets
\begin{align*}
\{(i,1)\} \quad \text{and} \quad \{(j,0):j\in [n]\smallsetminus \{i\}\}
\end{align*}
for every $i\in[n]$. See also Figure~\ref{fig:OLE}.

\begin{figure}[htp]
    \centering
    \input{tikz_O_and_CO}
    \caption{\small $\cO_{\mathsf{LE}}$ and $\pi(\cO_{\mathsf{LE}})$. The facet $\tau_1\in\cO_{\mathsf{LE}}$ is mapped to the subcomplex $\pi(\tau_1)\subseteq \pi(\cO_{\mathsf{LE}})$ that contains the edge $\{(2,0),(3,0)\}$ and the isolated node $\{(1,1) \}.$}
    \label{fig:OLE}
\end{figure}
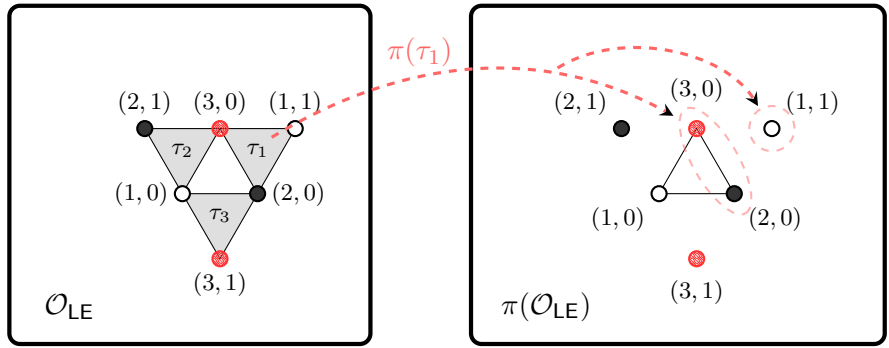

The second consistency-projection, $\tilde{\pi}$, is more specific and applies only to the realization complexes~$\cR(t)$, for $t\geq 1$. It is not using equality between values, but an equivalence relation between the vertices of $\cR(t)$, defined as follows. 
Given specific realizations $x_i \in \{0,1\}^t$ for the randomness of party~$i\in[n]$ at time~$t$, 
we say that nodes~$i$ and~$j$ are \emph{consistent} at time~$t$, denoted by 
\begin{equation}\label{eqn:consistent}
i \con{t} j
\end{equation}
if $K_{i}(t) = K_{j}(t)$. 
Note that once the randomness obtained by the parties up to round~$t$, $x_1,\ldots, x_n\in \{0,1\}^t$, are fixed, the event $K_{i}(t) = K_{j}(t)$ is deterministic (it has probability either~0 or~1). 
In the blackboard model, this equality  depends solely on the random bits received by nodes~$i$ and~$j$ during the first $t$ rounds. 
However, in the message-passing model, this equality also depends on the actual assignment of parties' port numbers.  
Also, it is worth observing that once $K_{i}(t) \ne K_{j}(t)$ in a specific instance of randomness, the two nodes $i$ and $j$ become inconsistent for the rest of the execution. However, they become aware of this fact only at the next round, where knowledge is exchanged (in both the blackboard and the message-passing model)\cite{moses-book2003}. 

Let $\rho = \{(i,x_i) :  i\in [n]\}$ be a facet of $\cR(t)$. We define the complex $\tilde{\pi}(\rho)$ as the complex on vertex-set $V(\rho)=\{(i,x_i) :  i\in [n]\}$ such that, for every non-empty $I\in[n]$, 
\begin{equation}\label{eq:pitilde}
\{(i,x_i) :  i\in I\} \in  \tilde{\pi}(\rho) \iff \forall (i,j)\in I\times I, \  i\con{t} j.
\end{equation}
The consistency complex captures all the possible relations of consistency for all possible generations of random strings. The consistency-projection $\tilde{\pi}$ applied simultaneously to all the facets of $\cR(t)$ results in the complex
\begin{equation}
\tilde{\pi}(\cR(t))=\bigcup_{\rho\in\cR(t)}\tilde{\pi}(\rho), 
\end{equation}
where the union is taken on the facets of~$\cR(t)$. Note that
$\tilde\pi(\cR(t))$ is a subcomplex of $\cR(t)$; its topological structure will be vital for our analysis.

\subsection{Randomized Solvability of Tasks Revisited}

A facet of~$\cR(t)$ is called a \emph{realization} of the system at time~$t$. 
By definition, there are $2^{nt}$ different realizations at time~$t$. 
Also, since each source of randomness generates a single bit uniformly at random at each round, we have 
\[
\Pr[R_i(1,\dots,t) = x]=2^{-t}
\]
for every $x \in\{0,1\}^t$.
It directly follows that the probability that a node~$i\in[n]$ receives random string $x\in\{0,1\}^t$ during the first $t$~rounds is $2^{-t}$, regardless to which randomness source node~$i$ is connected. However, there might be correlations between different nodes, whenever they are connected to the same randomness source (an information that is not given to the nodes a priori). 
Let $\rho=\{(1,x_1),\ldots (n,x_n) \}\in\cR(t)$ be a realization of the system at time~$t$, and let $\alpha\in\cA$ be a randomness-configuration of the system. We have 
\[
\Pr[\rho \mid \alpha]= 
\Pr\Big[\bigwedge_{(i,j)\in \alpha} R_j(1,\dots,t)= x_i\Big].
\]

We introduce a novel definition for solvability. This definition will be shown to be equivalent to the definition using facets of $\cP(t)$. 
\begin{definition}\label{def:solve_R}
A realization $\rho \in \cR(t)$ solves a symmetry-breaking task~$\cO$ if there exists~$\tau\in\cO$ and a name-preserving simplicial map 
$
\delta: \tilde{\pi}(\rho) \to \pi(\tau).
$
\end{definition}

Note that, in this definition, the map $\delta$ is not asked to be name-independent, 
since this property will be provided by the structure the projections impose.
Since every realization $\rho\in\cR(t)$ has some probability to occur at a given time, one can compute the probability of solving the task at time~$t$ by summing up the probabilities of the realizations that solve the task at time~$t$. 
For any $t\geq 1$, let $\cS(t)$ be the set of all the realizations 
of the system after $t$ rounds  that solves the task~$\cO$, and let us define the probability of $\cS(t)$ given some randomness-configuration $\alpha\in\cA$ as the sum of the probability of its facets, that is, 
\[
\Pr\left[\mathcal{S}(t) \mid \alpha \right] = \sum_{\sigma\in \cS(t)}\Pr\left[\sigma \mid \alpha \right].
\]
Again, by Kolmogorov's zero--one law, the limit when $t$ goes to infinity of the probability of $\cS(t)$ is equal to~0 or~1. Moreover, this new notion of solvability is equivalent to the (algorithmic) notion of solvability of Definition~\ref{def:solvabilityAS}.

\begin{lemma}
An input-free symmetry-breaking task $\cO$ is eventually solvable given a randomness-configuration~$\alpha \in\cA$ if and only if
\begin{equation}\label{eqn:solvability_by_cS(t)}
\lim_{t\to \infty} \Pr\left[\mathcal{S}(t) \mid \alpha \right] 
=1,
\end{equation}
where  $\mathcal{S}(t)$ is the set of all the realization $\sigma\in\cR(t)$ that  solve~$\cO$ at time~$t$.
\end{lemma}

\begin{proof}
First, assume $\sigma \in \cP(t)$ solves $\cO$. We show that $h(\sigma)\in \cR(t)$ solves $\cO$ for $h: \cP(t) \to \cR(t)$ the name-preserving simplicial map defined in Section~\ref{sec:isomorphism}. 
Fix $t$ and $\sigma\in\cP(t)$ that solves the task, and set $\rho = h(\sigma)$.
According to Definition~\ref{def:solve_P}, there is a name-preserving name-independent simplicial map $\delta:\sigma \to \tau$. where $\tau=\delta(\sigma)$ is a facet of~$\cO$. %

Given $\sigma$ and $\rho$, we can define a name-preserving simplicial map (in fact, an isomorphism) $\tilde h :\rho \to \sigma$ (being viewed as complexes) that for any 
$i\in[n]$ takes $(i,x_i)\in \rho$ to $(i,K_i)\in \sigma$; note that $\tilde h$ is the uniqe name-preserving simplicial map between $\rho$ and $\sigma$.

We claim that
$\lambda \triangleq \delta \circ \tilde h=\delta(\tilde h(\cdot))$ is a name-preserving simplicial map
$\lambda: \tilde\pi(\rho)\to \pi(\tau)$. 
Indeed,
(1)~$\lambda$~is name preserving: this follows immediately since $\tilde h,\delta$ are both name-preserving.
(2)~$\lambda$~preserves simplices: 
fix a realization $\rho=\{(i,x_i) : i\in[n]\} \in \cR(t)$ and let $K_i(t)$ be the knowledge of party~$i$ at time~$t$ given the realization~$\rho$. Then $\sigma=\tilde h(\rho)=\{(i,K_i(t)): i \in [n]\}$; in particular $\tilde h((i,x_i))=(i,K_i(t))$, since $\tilde h$~preserves names.
Let $\rho'$ be a facet in~$\tilde\pi(\rho)$. By the definition of~$\tilde\pi$, for any two vertices ${(i,x_i),(j,x_j)}\in \rho'$ it holds that $i\con{t} j$. Thus, by the definition of the $\sim$ relation, $K_i(t)=K_j(t)$ and thus 
$\{(i,K_i(t)),(j,K_j(t))\}\in \pi(\sigma)$. Since the consistency relation $\sim$ is transitive, the same argument holds for any subset of vertices in~$\rho'$. Hence,
$\lambda(\rho')$ is a simplex in~$\pi(\sigma)$.
We conclude that, if $\Pr[\sigma\mid\alpha]>0$, 
and $\sigma$ solves~$\cO$ by Definition~\ref{def:solve_P} via the map $\delta$,
then for $\rho=h(\sigma)$ and $\tau=\delta(\sigma)$ we have
$\Pr[\rho\mid\alpha]>0$, and $\rho$ solves $\cO$ by Definition~\ref{def:solve_R} via the map
$\lambda: \tilde\pi(\rho)\to \pi(\tau)$.

\medskip
Conversely, assume that $\rho\in \cR(t)$ solves $\cO$. We show that the facet $h^{-1}(\rho) \in \cP(t)$ solves $\cO$ for $h: \cP(t) \to \cR(t)$ the name-preserving simplicial map from Section~\ref{sec:isomorphism} (recall that $h$ induces an isomorphism on facets, thus it has a unique inverse for~$\rho$).
Let $\tau\in\cO$ and  $\delta:\tilde\pi(\rho)\to\pi(\tau)$ be given for $\rho$ by Definition~\ref{def:solve_R}.
Let $\sigma=h^{-1}(\rho)$ be given by the isomorphism. 
We claim that $\sigma$ solves $\cO$ via the map $\lambda = \delta\circ h$.
First, note that $\lambda:\sigma\to \tau$ is a name-preserving simplicial map, which follows since $h$ and $\delta$ are both name-preserving simplicial maps. This also implies that $\lambda(\sigma)=\delta(h((\sigma))=\tau$ is a facet of~$\cO$. 
Next, we argue that $\lambda$ is name-independent, namely, that any two parties with the same knowledge give the same output. 
Let $(i,K_i(t)),(j_,K_j(t))$ be two vertices in~$\sigma$, such that $K_i(t)=K_j(t)$. 
By definition, $i\con{t} j$. %
If so, then~$\{(i,x_i),(j,x_j)\}$ is a simplex in~$\tilde\pi(\rho)$, where $(i,x_i)=h((i,K_i(t)))$ and  $(j,x_j)=h( (j,K_j(t)) )$. 
This simplex
must be mapped by the simplicial map~$\delta$ to some simplex $\{(i,v),(j,v)\}\in \pi(\tau)$ (recall that simplices in $\pi(\tau)$ consist of nodes with similar outputs, Eq.~\eqref{eqn:pi}). 
Since the argument holds for arbitrary two nodes, it easily extends to any subset of nodes in~$\sigma$ that have identical knowledge, which proves that~$\lambda$ is a name-independent.
\end{proof}

In Figure~\ref{fig:overview} we illustrate the relations between the different complexes of this work.
\newboolean{SHOW_CP}
\setboolean{SHOW_CP}{false}     %

\begin{figure}[htb]
\begin{center}
\small
\begin{tikzpicture}
\matrix (m) 
[matrix of math nodes,row sep=3.5em,column sep=6em,minimum width=2em] 
{
\cR(t) & \cP(t) & \cO \\
|[name=cCR]|  
\displaystyle\bigcup_{\mathclap{\rho\in\cR(t)}}\tilde\pi(\rho)
& \ifthenelse{\boolean{SHOW_CP}}{
\color{gray!80!white}
\displaystyle\bigcup_{\mathclap{\sigma\in\cP(t)}}\pi(\sigma)
}{} 
&|[name=cCO]| %
\displaystyle\bigcup_{\mathclap{\tau\in\cO}}\pi(\tau)
\\
}; 
\path[-stealth]  
(m-1-2) edge node [above] {$\delta:\sigma\to\tau$} node [below] {\footnotesize (Def.~\ref{def:solve_P})} (m-1-3)
;
\path[-stealth] 
(m-1-3) edge node [auto] {$\pi$} (cCO)
(m-1-1) edge [stealth-stealth] node [auto] {$h$} (m-1-2)
(m-1-1) edge node [auto] {$\tilde\pi$} (cCR)
;
\path[-stealth]
(cCR) edge [bend right=45, shorten >= 2pt] 
    node [below] {\footnotesize (Def.~\ref{def:solve_R})}
    node [above=3pt] {$\delta:\tilde\pi(\rho)\to\pi(\tau)$}
    (cCO);

\ifthenelse{\boolean{SHOW_CP}}%
{
\tikzset{every path/.style={opacity=0.4}}
\path[-stealth]
(m-1-2) edge node [auto] {$\pi$} (m-2-2)
(cCR.east|-m-2-2) edge [stealth-stealth,shorten >=3pt] node [auto] {$h$} (m-2-2)    
;}%
{}
\end{tikzpicture}
\end{center}
\caption{\small A summary of our topological complexes and the relations between them.}
\label{fig:overview}
\end{figure}
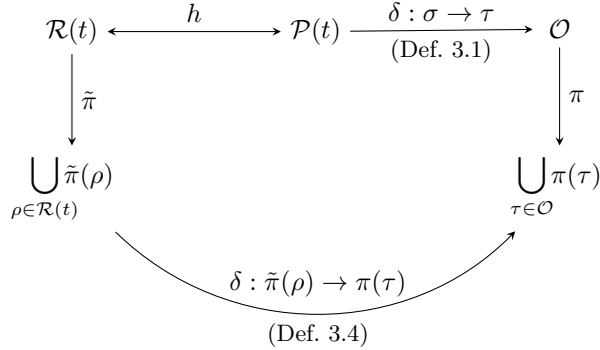

\section{Solvability of Leader Election via Topology}
\label{sec:le}

In the remainder of this paper we will consider the task of leader election, $\cO_\mathsf{LE}$. In this section we will discuss the conditions on randomness-configurations $\alpha\in\cA$ that make $\cO_\mathsf{LE}$ eventually solvable. We begin in Section~\ref{sec:blackboard_LE} with the blackboard model.
In Section~\ref{sec:port} we will consider the message-passing model.
Recall that the leader election task is defined by
$\mathsf{LE}=(\cI,\cO_\mathsf{LE},\Delta)$ with
$\cI = \{ (i,\bot):i\in[n]\}$ and $\Delta:\cI\to 2^{\cO_\mathsf{LE}}$ that maps the single facet of $\cI$ to the entire complex~$2^{\cO_\mathsf{LE}}$.
Further recall that  $\cO_\mathsf{LE}$ has facets 
\[%
\tau_i=\{(1,0),\dots,(i-1,0),(i,1),(i+1,0),\dots,(n,0)\},
\]%
for any $i\in[n]$,
and thus $\pi(\cO_\mathsf{LE}) = \bigcup_{\tau\in\cO}\pi(\tau)$
has  facets 
\[%
\{(i,1)\} \quad \text{and} \quad \{(j,0):j\in [n]\smallsetminus \{i\}\}
\]%
for any $i\in [n]$.
To ease readability, we will denote $\cO_\mathsf{LE}$
simply by $\cO$
from this point and on, as the task is clear from context.

\medskip
Intuitively, leader election is eventually solvable, in either model, if the algorithm can break symmetry, which in our topological view amounts to reaching some~$\rho\in\cR(t)$ with positive probability, such that $\tilde\pi(\rho)$ has an isolated vertex---that process will be the leader.
Impossibility is obtained when for any time $t$, for any~$\rho\in\cR(t)$ with positive probability $\tilde\pi(\rho)$ does not contain an isolated vertex and thus cannot be mapped to any $\pi(\tau)$ for $\tau\in\cO$.

In the blackboard  model (Section~\ref{sec:blackboard_LE}) we show that processes with the same randomness will always be connected in $\tilde\pi(\rho)$. Then, the only way to solve leader election is if there exists a process with its unique randomness source---eventually, the randomness will distinguish this node from any other node (with probability 1) so this node will be isolated in~$\tilde\pi(\rho)$ for any $\rho\in\cR(t)$ for which the realization of randomness that node obtained by time~$t$ differs form all the realizations of all other nodes.

In the message-passing model (Section~\ref{sec:port}), we show that one plus the dimension of any facet in~$\tilde\pi(\rho)$ is a multiple of~$g$, the GCD of $n_1,n_2,\ldots$, where $n_i$ is the number of processors connected to source~$\mathbf{R}_i$. 
Then, an isolated vertex $\tilde\pi(\rho)$ cannot exist unless the GCD is~1.
On the other hand, if the GCD is 1, we show an algorithm that leads to reducing the dimensions of certain facets 
(i.e., when looking on the evolution of $\rho\in\cR(t)$ as time goes by, that is, considering states $\rho'\in \cR(t')$ for $t'>t$ with $\Pr[\rho' \mid \rho]>0$)
until reaching a facet with dimension 0, that is, an isolated node.

\subsection{The Blackboard Model}
\label{sec:blackboard_LE}
With the above formulation we can determine the solvability of leader election in the blackboard model as a function of the specific configuration~$\alpha\in\cA$ of the system:
\begin{theorem}\label{thm:LE-BB}
Assume $k\le n$ distinct randomness sources are available to $n$ parties, where for any $i\in[k]$, there are exactly $n_i$ parties connected to~$\mathbf{R}_i$.
Then, leader election in the blackboard model 
is eventually solvable if and only if 
there exists a source~$i\in[k]$ such that $n_i=1$.
\end{theorem}
\begin{proof}

We prove the theorem separately for $k=1$ and for $k>1$. For each case we show that  %
Eq.~\eqref{eqn:solvability_by_cS(t)} 
holds if and only if $n_i=1$.

\medskip

\noindent
{\emph{Base case ($k=1$):}}
Let the randomness-configuration be such that $k=1$, that is, all the parties are connected to~$\mathbf{R}_1$ and see exactly the same stream of randomness. 
In particular, any realization that gives 
two different parties different randomness strings, has zero probability.
\begin{description}
\item[`if' direction:]
Since there is only a single source in the system ($\mathbf{R}_1)$, if~$n_1=1$ then the entire network contains a single party, $n=1$.
In this case leader election is trivial:
For any time~$t>0$, the complex
$\cR(t)$ has only 0-dimension facets: $\{(1,x_1)\}$ for any $x_1\in\{0,1\}^t$. 
For any $\rho\in \cR(t)$ it holds that  $\tilde\pi(\rho)= \rho$. %
Further, $\cO$ reduces to a single isolated node $\tau = \{(1,1)\}$, and thus $\pi(\cO)=\pi(\tau)=\{(1,1)\}$. 
It then follows that, for any $\rho=\{(1,x_i)\}\in \cR(t)$ 
there exists a name-preserving simplicial map $\delta:\tilde\pi(\rho) \to \pi(\tau)$, i.e., the map that takes $(1,x_i)$ to $(1,1)$. 

Thus, for every~$t$, it holds that $\cS(t)$ contains \emph{all} the facets of~$\cR(t)$, because any such realization solves the task $\mathsf{LE}$. 
For any~$t$, as well as in the limit~$t\to\infty$,
\(
\Pr\left[ \cS(t) \mid \alpha \right]
=1.
\)

\item[`only if' direction:]
For the other direction, let $\alpha'\in\cA$ be a randomness-configuration in which there is no source $i$ with $n_i=1$. 
Since there is only a single source ($k(\alpha')=1$) we have $n_1=n>1$ and we need to show that leader election is not eventually solvable.
Towards contradiction, assume that $\rho\in \cR(t)$ is a realization that solves LE, say, via a map $\delta: \tilde\pi(\rho) \to \pi(\tau_j)$ for some $j\in[n], \tau_j\in\cO$. 
By definition, any simplicial map~$\delta$
must preserve simplices. Let $(i,x_i)\in \tilde\pi(\rho)$ be the vertex  that is mapped to $(j,1)\in \pi({\tau_j})$. Since $\{(j,1)\}$ is a facet in~$\pi({\tau_j})$, the vertex $(i,x_i)$  must be isolated in~$\tilde\pi(\rho)$.

However, this  implies that $\Pr[\rho \mid \alpha']=0$, since $\alpha'$ dictates that all parties share the same randomness source and their views (randomness and blackboard content) are identical. 
Therefore, for any realization~$\rho'$ that has a non-zero probability (given~$\alpha'$), the projection $\tilde\pi(\rho')$ must be a single facet of dimension exactly $n-1$, that is, $\{(i,x) : {i\in [n]}\}$, for some value $x\in\{0,1\}^t$. On the other hand, $\tilde\pi(\rho)$ has a 0-dimension facet, and since $n>1$ its probability is zero given~$\alpha'$.

It follows that 
for any $t$, as well as in the limit,
$
\lim_{t\to \infty}\Pr\left[ \cS(t) \mid \alpha' \right]=0.
$
\end{description}

\noindent
\emph{General case ($k>1$):}
Assume we are given a randomness-configuration~$\alpha\in\cA$ with $k=k(\alpha)>1$ distinct randomness sources.
Without loss of generality assume
$0< n_1 \le n_2 \le \dotsb \le n_k$.

\begin{description}
\item[`if' direction:]
Assume $\alpha\in\cA$ satisfies $n_1=1$. 
At any time~$t$ there are exactly $2^{kt}$ unique realizations $\rho \in \cR(t)$ with positive probability conditioned on~$\alpha$, and recall that all such realizations are equiprobable  (Lemma~\ref{lem:equiprob-states}).
Let us define 
\begin{align*}
\cS_1(t) \overset{\text{def}}{=} 
\Big\{ \rho= \{(1,x_1),\ldots,(n,x_n)\} \in \cR(t) \ \Big\vert\  
\forall i>1, x_i\ne x_1\Big \}
\end{align*}
to be the set of all the realizations $\rho\in\cR(t)$  
in which the randomness obtained by the first party is unique, $x_1 \ne x_i$ for all $i>1$. 
There are $2^t\cdot (2^t-1)^{k-1}$ such realizations with positive probability given~$\alpha$.
Note that each such realization solves the task $\mathsf{LE}$ via the (unique) name-preserving map
$\delta:\tilde\pi(\rho) \to \pi(\tau_1)$, 
thus $\cS_1(t) \subseteq \cS(t)$.
Since all the positive-probable realizations are equiprobable we have,
\begin{align*}
\Pr \left[\cS(t) \mid \alpha \right]
&\ge 
\Pr \Big [\bigcup_{\rho \in \cS_1(t)} \rho \mid  \alpha \Big] 
= 
2^t (2^t-1)^{k-1} \cdot 2^{-kt} 
= \frac{(2^t-1)^{k-1}}{2^{t(k-1)}}\\
& \ge 1- \frac{(k-1)2^{t(k-2)}}{2^{t(k-1)}} 
= 1-\frac{k-1}{2^t}.
\end{align*}
From the above, $\lim_{t\to\infty}\Pr \left[\cS(t)\mid \alpha \right] = 1$ as required.

\item[`only if' direction:] 
Assume a randomness-configuration~$\alpha'$ (with $k(\alpha')>1$) in which $n_1>1$ (hence, there exists no $i$ with $n_i=1$).
Let us fix a time~$t$, and let $\rho \in \cR(t)$ be a realization that solves the task $\mathsf{LE}$ 
via~$\delta_j: \tilde\pi(\rho) \to \pi(\tau_j)$. 
In order for $\delta_j$ to be a name-preserving  simplicial-map, it is required that $(j,x_j)\in V(\tilde\pi(\rho))$ is isolated in~$\tilde\pi(\rho)$, since the vertex with name~$j$ is a 0-dimensional facet of~$\pi(\tau_j)$. 

On the other hand, let $\rho'\in \cR(t)$ be a realization with positive probability given $\alpha'$, namely, $\Pr[\rho'\mid\alpha']>0$.
Since in~$\alpha'$ we have that $\forall i, n_i>1$, for any party~$j$ there must exist another  party $j'\ne j$ that is connected to the same randomness source as~$j$.
Hence, for any time~$t$, the randomness $j$ and $j'$ see is identical, $x_j = x_{j'}$.
Furthermore, in the blackboard model, equality of randomness is equivalent to equality of knowledge, since the knowledge of a party is just its randomness along with the content of the blackboard. 
Therefore, regardless of the realizations $\{x_i\}_{i\notin\{ j,j'\}}$ of the other parties, it holds that $j\con{t} j'$  and 
$\{(j,x_j),(j',x_{j'})\}$ is a simplex in~$\tilde\pi(\rho')$.
So, for any $\rho'$ with $\Pr[\rho'\mid\alpha']>0$ and any party~$j$, we get that $(j,x_j)$ is \emph{not} isolated in~$\tilde\pi(\rho)$.

These two arguments imply that for any facet $\rho\in \cR(t)$ that solves the task $\mathsf{LE}$ it holds that
\(
\Pr[\rho \mid \alpha']=0.
\)
If we let $\cS(t)$ denote all the realizations the solve leader election at time~$t$, the above proves that for any~$t$, as well in the limit,
\[
\Pr \left[\cS(t) \mid \alpha' \right]=  0.
\]
In particular, leader election is not eventually solvable in this case.
\qedhere
\end{description}
\end{proof}

\subsection{The Message-Passing Model}\label{sec:port}
We now turn to the message-passing model where nodes are indistinguishable and connected by point-to-point channels as a clique.
Each node (privately) labels its $n-1$ neighbours with unique labels from $\{1,\ldots, n-1\}$ 
in an arbitrary way, which we refer to as the node's port-numbers. 
Then, when a node sends a message to some port~$p$ it always reaches the same party; however, if a different node sends a message to (its own) port~$p$, it might reach a different party according to that node's private labeling.
The main difference from the blackboard model is that a party's knowledge might be affected by its port-numbering in addition to its randomness. Unlike the blackboard model, where similar randomness means similar knowledge, here parties may have different knowledge while observing the same randomness (however, if their randomness is different, their knowledge will be different as well).

We recall that the numbering of the ports is arbitrary. In the following we ask which configurations lead to solving leader-election \emph{regardless} of the specific port numbers.
Alternatively, we ask which configurations prevent any protocol from solving
leader-election for at least one port-numbers labeling. 
That is, given a randomness-configuration, we look at the \emph{worst case} for setting the port numbers and ask whether or not leader-election is eventually solvable. We term this question \emph{worst-case leader election}.

The following Theorem~\ref{thm:le-PN}, which is this section's main theorem, shows that \emph{worst-case} eventual solvability of leader election depends on the greatest common divisor (GCD) of the number of parties connected to each randomness source.
\begin{theorem}\label{thm:le-PN}
Assume $k\le n$ distinct randomness sources are available to $n$ parties, where for any $i\in[k]$, exactly $n_i \ge 1$ parties are connected to~$\mathbf{R}_i$. 
Worst-case leader election is eventually solvable if and only if 
$
\gcd(n_1,n_2,\ldots,n_k)=1.
$
\end{theorem}

We start by showing the following technical lemma that explains how the consistency-projected complex $\tilde\pi(\cR(t))$ changes with time. 
\begin{lemma}\label{lem:dimension}
For any $\alpha\in\cA$ such that
$\gcd(n_1,\ldots,n_k)=g$ there exists a way to number ports so that 
for any realization~$\rho\in\cR(t)$ that has positive probability $\Pr[\rho \mid \alpha]>0$, 
the dimension~$\dim(\gamma)$ of any facet $\gamma\in \tilde\pi(\rho)$ satisfies 
$
g \mid \dim(\gamma)+1.
$
\end{lemma}

Note that the above suggests that worst-case leader election is not eventually solvable  when the GCD is larger then one, proving the ``only if'' direction of Theorem~\ref{thm:le-PN}.
\begin{corollary}\label{cor:le-imp-PN}
Assume $\alpha'\in\cA$ such that $\gcd(n_1,\ldots,n_k)=g$ and $g>1$, then there exists a way to assign port-numbers to channels so that leader election is impossible.
\end{corollary}
\begin{proof}
From Lemma~\ref{lem:dimension} we know that there exists a way to number ports such that 
any realization~$\rho\in\cR(t)$ for which~$\Pr[\rho \mid \alpha']>0$, has facets of dimension at least~$g-1 > 0$. In particular, there is no isolated vertex in~$\tilde\pi(\rho)$, as this will imply a facet~$\gamma$ with $\dim(\gamma)=0$, but $g \nmid \dim(\gamma)+1$, which is a contradiction. 
Since there is no isolated vertex in~$\tilde\pi(\rho)$, there exists no simplicial map from~$\tilde\pi(\rho)$ to $\pi(\tau_j)$ for any facet~$\tau_j\in\cO$ and thus $\rho$ does not solve leader election. 
We get that for any $t$, no realization that solves leader election has positive probability, $\Pr\left[\cS(t)  \mid \alpha' \right]=0$.
\end{proof}

\begin{proof}[Proof of Lemma~\ref{lem:dimension}.]
Let $\alpha\in\cA$ be a randomness-configuration in which $\gcd(n_1,\ldots,n_k)=g$. 
Split the $n$~parties into $g$ disjoint subsets of nodes where subset $i\in \{1,\ldots,g\}$ holds exactly $n_j/g$ parties which are connected to~$\mathbf{R}_j$ (for all $j\in[k]$). 
For this part only rename the parties as
$0,1,\ldots, n-1$ where the first $n_1$ parties are connected to the first source and the next $n_2$ parties are connected to the second one, etc.
We assign the $j$-th port ($j\in [n-1]$) of party~$i\in \{0,1,\ldots, n-1\}$ 
to be connected to party number
\[
\big( (i+j)\text{mod }g  + \lceil i/g \rceil \cdot g + \lceil j/g \rceil \cdot g\big)  \mod n.
\]
There exists an isomorphism~$f:\{0,\dotsc,n-1\}\to \{0,\dotsc,n-1\}$ that takes any party $i\in \{0,1,\ldots, {n-1}\}$ written as $i=r+m\cdot g$ for $r<g$, $m\in\mathbb{N}$, to the party $f(i)=(r+1 \mod g) + mg$. This isomorphism preserves the assignment of randomness source and it preserves port numbers. 
Thus, fixing any realization of randomness $(x_1,\ldots, x_n)$ at round~$t$, for any party~$i$ we get that $i \con{t} f(i)$: 
they are connected to the same source and thus have the same randomness ($x_i=x_{f(i)}$). Moreover, every message that party~$i$ receives from port~$j$, is also received by~$f(i)$ from its own port~$j$, due to the way we number ports.
This can be seen by induction; it trivially holds for round~$0$; now assume it  holds for round $t-1$. Note that if $i$ is connected to $p$ at its port $j$, then $f(i)$ is connected to $f(p)$ in its port~$j$; furthermore, if $p$ sees $i$ in its $j'$ port, then $f(p)$ sees $f(i)$ it its $j'$ port. Since $p,f(p)$ hold the same information at round $t-1$ by the induction hypothesis, $i$ and $f(i)$ will receive the same information from $p,f(p)$ respectively, and the claim will hold for time~$t$ as well. 
Let us denote the equivalence class of party~$i$ by
\begin{align*}
\llbracket i \rrbracket = \{ f^{(c)}(i) \mid c\in[g] \}.
\end{align*}
where $f^{(c)}=f(f(\ldots))$ denotes the $c$-th composition of $f$ with itself, and $f^{(0)}$ is the identity function (also note that $f^{(g)}=f^{(0)}$).
Let~$\rho \in \cR(t)$ be a realization such that  $\Pr[\rho \mid \alpha]>0$, and let $\gamma \in \tilde\pi(\rho)$ be a facet. 

\begin{claim}\label{clm:cylinder}
If $(i,y_i)\in \gamma$, then for any $j\in\llbracket i \rrbracket$ we have that $(j,y_i) \in \gamma$. 
\end{claim}

\noindent\textit{Proof of claim.}
Suppose $(j,y_i)\not\in \rho$.  Then there exists $y_j\ne y_i$ such that $(j,y_j)\in \rho$ which implies  that $\Pr[\rho \mid \alpha]=0$ since $i$ and $j$ are connected to the same randomness source in~$\alpha$ according to the way we numbered parties and the definition of~$f()$.
Therefore, $(j,y_i)$ is a node of~$\tilde\pi(\rho)$, and we are left to show that $(j,y_i)\in \gamma$.
As argued above, we know that, for any time $t$, 
$i \con{t} f(i)$ and thus (for the realization~$\rho$)
\[
i \con{t} f(i) \con{t}  \dotsb \con{t} f^{(g-1)}(i).
\]
Since the $\con{t}$ relation is transitive, all these vertices are consistent with one another. 
In particular, note that $i \con{t} j$ since $j\in\llbracket i\rrbracket$.
Further, for any vertex $(x,y)\in \gamma$, we have that
$x\con{t} i$ by the definition of the projection~$\tilde\pi$ (Eq.~\eqref{eq:pitilde}) and the fact that $\gamma\in\tilde\pi(\rho)$.
Now, by the transitivity of the consistency operator 
and the fact that $i \con{t} j$,
we get that also $x \con{t} j$. 
Finally, since $\gamma$ is a facet of~$\tilde\pi(\rho)$, the above implies that $(j,y_i)\in \gamma$, 
which completes the proof of Claim~\ref{clm:cylinder}. 
\hfill$\diamond$

\medbreak

Finally, we conclude the proof of the lemma. Let $\gamma\in \tilde\pi(\rho)$ be a facet.  For any $(i,y_i)\in \gamma$, Claim~\ref{clm:cylinder} proves that also $(j,y_i)\in \gamma$ for any~$j\in\llbracket i\rrbracket$. Since the sets $\llbracket i\rrbracket$  form a partition of~$\{0,...,n-1\}$ and since $|\llbracket i\rrbracket|=g$ for any~$i$, we conclude that
$g\mid \dim(\gamma)+1$, which completes the proof of Lemma~\ref{lem:dimension}.
\end{proof}

\medskip
We now move to proving the other direction of Theorem~\ref{thm:le-PN}, that is, showing that if the GCD of the size of subsets that are assigned the same randomness source equals one, then leader election is solvable, \emph{regardless} of the specific port numbers. 
Towards this goal, we show the following technical lemma that
once again considers the structure of the projected complex $\tilde\pi(\cR(t))$ over time. The lemma suggests that facets in $\tilde\pi(\cR(t))$ (that stem from realizations with positive probability) eventually ``split'' into smaller facets. 
The change in their dimension is always a multiple of the GCD.

\begin{definition}
We say that a realization $\rho'\in\cR(t')$ \emph{succeeds} a realization $\rho\in\cR(t)$ and denote $\rho \prec \rho'$ if (i) $t'>t$ and (ii) for any $i\in[n]$, if $(i,x_i)\in \rho'$ then $(i,x_i(1,\ldots,t))\in \rho$.
\end{definition}

\begin{lemma}\label{lem:dimReduce}
Fix a randomness-configuration~$\alpha\in \cA$.
Assume a realization~$\rho\in\cR(t)$ with $\Pr[\rho \mid \alpha]>0$.
Let $\gamma_1,\gamma_2\in \tilde\pi(\rho)$ be two facets, 
where without loss of generality we assume  
$\dim(\gamma_1) \le \dim(\gamma_2)$.
For any $t'>t$, let $\Gamma(t')$ 
be the set of realizations $\rho'\in\cR(t')$ that satisfy the following three conditions:
(1) $\rho \prec \rho'$; 
(2) $\Pr[\rho' \mid \rho,\alpha]>0$ and 
(3) if  $\gamma$ is the largest facet in $\rho'$ that contains a node with name from $\textsf{names}(V(\gamma_2))$, then   
$\dim(\gamma)\le \max\{\dim(\gamma_1), \dim(\gamma_2)-\dim(\gamma_1)-1\}$.
It holds that
\[
\lim_{t'\to\infty} \Pr\left[
\Gamma(t')
\ \middle\vert\  \alpha, \rho \right] = 1.
\]
\end{lemma}

\begin{proof}
Let $V_1 = \mathsf{names}(V(\gamma_1))$ and $V_2 = \mathsf{names}(V(\gamma_2))$ be the names of all nodes in the respective set.
Note that if at each round, parties send all their information to each other, this information can
simulate any randomized protocol the network can run, in the sense that the parties can output the same output of any other protocol.
In particular, it simulates the $\textsc{CreateMatching()}$ procedure depicted in Algorithm~\ref{alg:matching}, that essentially creates a matching between~$V_1$ and~$V_2$. The matching is performed as follows: each party in~$V_1$ randomly picks a node in~$V_2$, and sends a message to that node, asking to be matched with it. If a node in~$V_2$ received only a single matching-request message, it accepts it. If it received more than a single request, it accepts just one request and rejects the others. This process continues with the remaining unmatched nodes until all the nodes in~$V_1$ are matched (assuming $|V_1| \le |V_2|$).

\begin{algorithm}[ht]  
\caption{The CreateMatching Procedure}
\label{alg:matching}
\begin{algorithmic}[1]
\small

\Statex \textbf{Input:} $n$ identical parties, connected as a clique, where $n_1$ parties connected to $\mathbf{R}_i$ and  $n_2$ parties connected to $\mathbf{R}_j$ with $j\ne i$, and $\gcd(n_1,n_2)=1$.

\Statex 
\State \textbf{Initially:} $V_1$ is the set of all nodes connected to~$\mathbf{R}_i$ and $V_2$ is the set of all parties connected to~$\mathbf{R}_j$. 
The procedure assumes that this separation is already known to all the participating parties.
Ignore any other parties in the network, if they exist.

\Statex 
\Procedure{CreateMatching}{$V_1,V_2$}	
\Comment{Assuming (without loss of generality) $|V_1| \le |V_2|$}
\State All nodes in $V_1,V_2$ set themselves as \texttt{active}.
\Repeat
\State \LongText{Each \texttt{active} node in $V_1$ randomly picks an \texttt{active} 
neighbour from $V_2$ and sends a messages to that selected neighbour.}  
\label{step:matching:s1}
\State \LongText{Each node in $V_2$ that have received at least one message,  selects the minimal port from which a message has arrived, and sends an ACK message to that origin node (in $V_1$).}
\label{step:matching:s2}
\State \LongText{Nodes from $V_1$ that received an ACK message in the previous step set themselves to \texttt{done} and broadcasts this event to all their neighbours. Nodes in $V_2$ that \emph{sent} an ACK message in the previous step, set themselves to \texttt{done} and broadcast this event to all their neighbours.}
\label{step:matching:s3}
\Until{all $V_1$ nodes are \texttt{done}} \label{line:repeat_done}
\EndProcedure

\end{algorithmic}
\end{algorithm}

\begin{lemma}\label{lem:matching}
Let $n$ identical nodes be connected as a clique where $n_1$ ($n_2$) nodes are connected to the randomness source~$\mathbf{R}_i$ ($\mathbf{R}_j$, $j\ne i$). 
At the end of \textsc{CreateMatching()}, there exists a matching~$M$ between all the $n_1$ parties connected to~$\mathbf{R}_i$ and $n_1$ parties connected to~$\mathbf{R}_j$, and every party outputs whether it is matched or not.
\end{lemma}
\begin{proof}
The proof is rather straightforward. Call $V_1$ ($V_2$) the set of parties connected to $\mathbf{R}_i$ ($\mathbf{R}_j$). 
Consider the state after the first time the procedure reaches Line~\ref{step:matching:s1}. 
At this point, every node in $V_1$ has selected a single node in~$V_2$, which we consider as throwing $|V_1|$~balls into $|V_2|$~bins. 
In Line~\ref{step:matching:s2} each non-empty bin selects exactly a single ball and ``ignores'' the other balls in it. 
Since each ball is uniquely identified with a node in~$V_1$, this creates a matching between all the non-empty bins (nodes) in~$V_2$ to some nodes in~$V_1$. Nodes that belong to the matching now become \done, and the matching continues recursively with the remaining nodes. 
Since in every iteration there must be at least one non-empty bin, every iteration increases the size of the matching by at least one. Eventually, all the parties in~$V_1$ will be matched (since $|V_1|\le |V_2|$ by our assumption).
At this point, all the parties that are \done participate in the matching and all the remaining \ctive parties are unmatched. Note that all the parties learn when \textsc{CreateMatching()} has terminated 
(e.g., by counting the number of \done parties in~$V_1$ when reaching Line~\ref{line:repeat_done}).
\end{proof}

After the execution of \textsc{CreateMatching()}, the nodes whose names belong to $V_2$ has split into two subsets, $V_2 = V_m \cup V_{um}$ of ``matched'' and ``un-matched'' nodes, of sizes $|V_m|=n_1$ and $|V_{um}|=n_2-n_1$, respectively.
Let $t'$ be a time after this split happens, and let $\rho'$ be the realization that corresponds this execution, so $\rho \prec \rho'$.
Note that the knowledge of nodes in these two subsets must be different\footnote{I.e., their state $\in \{\mathsf{matched},\mathsf{unmatched}\}$ is function of their knowledge, hence their knowledge must be different in order to obtain different state.}.
By the definition of the projection~$\tilde\pi$, no facet in $\tilde\pi(\rho')$ can contain nodes from both $V_m$ and $V_{um}$, as they are inconsistent. 
Hence, the maximal dimension of any facets of $\tilde\pi(\rho')$ that contain some nodes with name from~$V_2$ is at most $\max\{n_1-1, n_2-n_1-1\}$. 
Note that \textsc{CreateMatching()} requires that the participating parties know their partition into $V_1,V_2$. 
This holds in our case since $\gamma_1,\gamma_2$ are two different facets and parties that belong to different facets have different knowledge (and in a single rounds they can be aware of this).
\end{proof}  %

With the above lemma in hand, we can prove the `if' direction of Theorem~\ref{thm:le-PN}. 
The intuition is that 
any realization~$\rho$ with positive probability whose projection $\tilde\pi(\rho)$ has a facet of dimension~$d\ge 1$ is eventually succeeded by $\rho'$ whose projection~$\tilde\pi(\rho')$ has facets with maximal dimension strictly less than~$d$.
At the beginning of the computation, each set of parties that are connected to the same randomness source, are consistent in their knowledge\footnote{To be more precise, the knowledge already might be different due to differences in parties' port numbers. This can only help us in reaching a leader. In the discussion we ignore this option and assume that knowledge differences only come from the randomness and its affect on being matched or unmatched.}, and form a facet in the projection of the relevant realization. The fact that the GCD of the sizes of these subsets is 1 implies that we can reduce the dimension of the (largest) facet again and again until we reach facets of dimension~0.

To illustrate the above in an algorithmic manner, if we start with two sets of sizes $n_1<n_2$, we can perform a matching between the sets, and turn off all the nodes that belong to $V_m$ in the matching. This leaves us with two new sets of parties, $V_1$ and $V_{um}$ of sizes $a=n_1, b=n_2-n_1$, respectively, where parties that belong to the same set connect to the same randomness source. 
Repeating this process again and again yields a subset of size $\gcd(n_1,n_2)$. The fact that $\gcd(n_1,n_2,\ldots,n_k)=1$ and that $\gcd()$ is an associative function, $\gcd(a,b,c)=\gcd(\gcd(a,b),c)$, means we can repeat the above process on different subsets of parties (in a similar way to the Euclidean algorithm) until the dimension of the maximal facet becomes~$0$.

\begin{proof}[Proof of Theorem~\ref{thm:le-PN}, `\textbf{if}' direction]

We just need to show that for any $\alpha\in\cA$ with $\gcd(n_1,n_2,\ldots,n_k)=1$ there exists a realization $\rho\in \cS$ with positive probability given $\alpha$. Then, by Kolmogorov's zero--one  Eq.~\eqref{eqn:solvability_by_cS(t)} holds.
We will assume, without loss of generality, that at the onset of the analysis shown here, $\tilde\pi(\cR(t))$ contains facets of dimensions $n_1-1,n_2-1,\ldots,n_k-1$. This can easily be achieved by letting the parties exchange their randomness until $k$ differences appear (if $k$ is known), or until the parties' randomness distinguish $k'$ subsets whose sizes' GCD is 1 (if $k$ is unknown, restarting the protocol whenever future randomness exhibits new differences among parties that previously belonged to the same set). This step eventually succeeds with probability 1. We set time~$0$ to be after this step has completed, and below consider only $t>0$.
Let $\alpha\in\cA$ be a randomness-configuration with $k=k(\alpha)$  in which $\gcd(n_1,n_2,\ldots,n_k)=1$.
Let 
\[
\cS^*(t) = \{ \rho \in \cR(t) \mid \dim(\tilde\pi(\rho))=0\}
\]
be the set of realizations whose consistency-projection has dimension~0, that is, it is composed of $n$ isolated vertices.
Clearly any realization in $\cS^*$ solves leader election, $\cS^*\subseteq \cS(t)$, so finding a realization in~$\cS^*$ with positive probability given~$\alpha$ suffices to complete the proof. 
Assume a realization~$\rho\notin \cS^*(t)$ with positive probability, $\Pr[\rho\mid \alpha]>0$.
Let $d_1,d_2,\ldots$ be the dimensions of $\tilde\pi(\rho)$'s facets. Since $\rho\notin\cS^*(t)$ we know that $d=\max(d_1,d_2,\ldots) > 0$.
We argue that  
$\gcd(d_1+1,d_2+1,\ldots)=1$. 
This will follow from the following lemma
\begin{lemma}
For any $\alpha\in\cA$, $t'>t$ and realizations $\sigma\in\cR(t), \sigma'\in\cR(t')$,
such that $\sigma \prec \sigma'$,
the (unique) name-preserving map $\delta:\tilde\pi(\sigma')\to\tilde\pi(\sigma)$ is a simplicial map.
\end{lemma}
\begin{proof}
Let $\tau'\in\tilde\pi(\sigma')$ be a facet. By the properties of~$\tilde\pi$ we have that for any two nodes $(i,x_i),(j,x_j) \in V(\tau')$ we have $i\con{t'} j$, with respect to the realization~$\sigma'$. 
Note that $(i,x_i(1,\ldots,t)),(j,x_j(1,\ldots,t))\in V(\sigma)$ since $\sigma\prec \sigma '$. 
It thus also holds that 
$i \con{t} j$ with respect to the realization~$\sigma$
since knowledge is cumulative (i.e., if the knowledge in time~$t$ makes them inconsistent, this holds also for any time $t'>t$).
Thus, $\delta(\tau')$ is a simplex in~$\tilde\pi(\sigma)$.
\end{proof}

Assume, towards contradiction, that $\gcd(d_1+1,d_2+1,\ldots)=g >0$, and let $\sigma_1,\sigma_2,\ldots,\sigma_t=\rho$ be the sequence of realizations that correspond to the computation, $\sigma_i\prec\sigma_{i+1}$ for any $i<t$.
Note that due to the above lemma, facets in $\tilde\pi(\sigma_{t-1})$ have dimension $\in \{ (-1)+\sum D \mid D\in 2^{\{d_1+1,d_2+1,\ldots\}} \}$ (where a sum of a set is the sum of the elements in the set), since facets in $\tilde\pi(\sigma_{t-1})$ must be composed of one or more facets of $\tilde\pi(\sigma_t)$. 
Let $d'_1,d'_2,\ldots$ be the dimensions of the facets of~$\tilde\pi(\sigma_1)$. The above implies that $\gcd(d'_1+1,d'_2+1,\ldots)=g >0$, which is a contradiction.
Therefore, unless all the facets of~$\tilde\pi(\rho)$ are of dimension~$0$, there must exist a facet~$\gamma_1$ with dimension strictly less than~$d$ (otherwise, if all facets have dimension~$d$, then $\gcd(d_1+1,d_2+1,\ldots)=d+1>1$). Now, we apply Lemma~\ref{lem:dimReduce} on $\rho$, where $\gamma_1$  is as defined above (i.e., the fact with $\dim(\gamma_1)<d$), 
and $\gamma_2$ is any facet with maximal dimension~$\dim(\gamma_2)=d$.
The lemma suggests that there exists a time $t'$ and a realization $\rho'\in\cR(t')$ that succeeds $\rho$, $\Pr[\rho' \mid \rho,\alpha]>0$, where the maximal facet in~$\tilde\pi(\rho')$ has dimension strictly less than~$d$, that is, $\dim(\tilde\pi(\rho'))<d$.\footnote{Note that multiple applications of Lemma~\ref{lem:dimReduce} might be needed, e.g., when multiple facets in~$\tilde\pi(\rho)$ are of dimension~$d$.}
Note that,
\begin{align*}
\Pr[ \rho' \mid \alpha] \ge \Pr[\rho'\mid \alpha,\rho]\Pr[\rho\mid \alpha] >0.
\end{align*}
Hence, $\rho'$ is a realization with positive probability, where facets in~$\tilde\pi(\rho')$ have dimension strictly less than the ones in~$\tilde\pi(\rho)$ we began with. 
We can now apply the same reasoning on~$\rho'$ and keep reducing the dimension of the projected complex, until we reach a realization~$\rho^*$ with $\Pr[\rho^* \mid\alpha]>0$, for which all facets of~$\tilde\pi(\rho^*)$ are of dimension~$0$. 
\end{proof}

\section{Conclusion}

This paper focuses on input-free symmetry-breaking tasks, and presents a topological framework for studying the solvability of such tasks by \emph{randomized} algorithms. The applicability of this framework was demonstrated by studying the solvability of leader election in environments where correlations may exist between the randomness sources assigned to the processing nodes. Thus, the paper is expected to help move the study of randomized algorithms under the umbrella of algebraic topology. 

Our analysis resulted in a complete characterization of the solvability of leader election by randomized algorithms. 
In Appendix~\ref{app:ni-task} we show that the same conditions required for solving leader election also suffice for solving any \emph{name-independent} task (but they are not necessary).
Extending 
this work to \emph{any} task $(\cI,\cO,\Delta)$ is an appealing research direction. 
A first step may consist of extending this paper's framework to input-free tasks for which the output complex~$\cO$ is not symmetric. An intriguing example is electing a leader and a deputy leader (where the latter is to be used as an immediate backup in case the leader fails), under the constraint that some nodes may only be leaders, some nodes may only be deputy leaders, some nodes may be either of the two, and some nodes may be neither. 
Another compelling direction is extending the communication model to networks with arbitrary structure.

\section*{Acknowledgements}
The authors would like to thank Ami Paz for fruitful discussions and insightful explanations on topology in distributed computations.
P.~Fraigniaud is supported in part by ANR projects DESCARTES and FREDDA. 
R.~Gelles is supported in part by ISF grant 1078/17.

\bibliographystyle{alphabbrv-doi}
\bibliography{rand}

\appendix

\section{Algebraic Topology: Basic Definitions}
\label{app:topology101}

We give here a brief survey of the topological terms used in this paper, and refer the reader to~\cite{Book-topology} for a more complete treatment of the subject.

An \emph{abstract simplicial complex}~$\cK$ is a nonempty set of sets (simplices) $\cK=\{\sigma_i\}_i$
and it holds that if $\sigma$ is a simplex then any non empty subset $\rho\subseteq \sigma$ is also a simplex, $\rho\in\cK$.
The elements of a simplex~$\sigma$ are called \emph{nodes} or \emph{vertices}, and are denoted by $V(\sigma)$. The set of all nodes, $V(\cK)=\bigcup_{\sigma\in\cK} V(\sigma)$ is called the \emph{node-set} of the complex.
The dimension of a simplex is $\dim(\sigma)=|V(\sigma)|-1$.
In particular, a single node $\sigma=\{v\}$ has dimension~$0$.
A \emph{facet} of~$\cK$ is a simplex that is not contained in any other simplex of~$\cK$. Note that the set of facets fully defines the complex. A facet of dimension~$0$ is called an \emph{isolated} node.
The dimension of a complex is the maximal dimension of its facets. A complex whose all facets have the same dimension is called \emph{pure}.

For two complexes, $\cK$ and $\cL$, we say that $\cK$ is a subcomplex of~$\cL$ if $\cK\subseteq\cL$.
For a set $X\subseteq V(\cK)$, the \emph{induced complex of~$\cK$ on~$X$} is the complex $\{\sigma\in\cK \mid V(\sigma)\subseteq X\}$.
A vertex map from~$\cK$ to~$\cL$ is any function $f: V(\cK)\to V(\cL)$. A \emph{simplicial map} $\delta:\cK\to\cL$ is a vertex map such that for any $\sigma\in\cK$ it holds that $\delta(\sigma)=\{\delta(v)\mid v\in \sigma\}\in\cL$, that is, it maps simplices in~$\cK$ to simplices in~$\cL$ (i.e., it preserves simplices).
Complexes $\cK$ and $\cL$ are said to be isomorphic if there exist simplicial maps $f:\cK\to\cL$ and $f^{-1}:\cL \to \cK$, such that for any $\sigma\in\cK$, $\sigma=f^{-1}(f(\sigma))$, and for any $\rho\in\cL$, $\rho=f(f^{-1}(\rho))$.
A \emph{chromatic} complex~$\cK$ is a complex augmented with a naming function $\mathsf{name}: V(\cK) \to C$ where $C$ is called the set of names (colors). A vertex map~$f: V(\cK)\to V(\cL)$ \emph{preserves names} 
if $\cK$ and $\cL$ are chromatic, and for any $v\in V(\cK)$ we have $\mathsf{name}(v)=\mathsf{name}(f(v))$.
In this paper all complexes are chromatic and all maps are name-preserving.

\section{Technical Lemmas}

The following lemma shows that, for a given time and randomness-configuration, all the global states with positive probability are equiprobable.
\begin{lemma}\label{lem:equiprob-states}
Given $t>0$ and $\alpha\in\cA$ where exactly $k=k(\cA)$ different randomness sources are connected to the parties in~$\alpha$,
define the set of $\alpha$-inconsistent randomness,
\[
\textsf{B}_\alpha = \Big\{ (x_1,\ldots,x_n)\in (\{0,1\}^{t})^{n} 
 \ \Big\vert\   \exists i,j\in[n], c\in[k] \text{ s.t. } x_i\ne x_j \text{ but } (i,c),(j,c)\in \alpha \Big\}.
\]
For any facet $\sigma\in\cR(t)$ 
with nodes $V(\sigma)=\{(i,x_i) \mid i \in [n]\}$,
\[
\Pr[\sigma \mid \alpha] =
\begin{cases}
 0  &   ( x_1,\dotsc, x_n) \in \textsf{B}_\alpha \\
 2^{-tk} & ( x_1,\dotsc, x_n) \notin \textsf{B}_\alpha
\end{cases}.
\]
\end{lemma}
The proof follows directly from the definition since $\Pr[\sigma \mid \alpha] = \Pr[(R_1,\ldots,R_n)(1,\ldots,t)=(x_1,\ldots,x_n) \mid \alpha]$.
Recalling that the mapping~$h$ defined in Section~\ref{sec:isomorphism} induces a name-preserving isomorphism on the facets of $\cP(t)$ and $\cR(t)$,
we conclude:
\begin{corollary}
Given $t>0$ and $\alpha\in \cA$ with exactly $k=k(\cA)$ randomness sources connected to the parties,
for any $\sigma \in \cP(t)$ we have
\[
\Pr[ \sigma \mid \alpha] = \Pr[ h(\sigma) \mid \alpha]  \in \{0, 2^{-tk}\}.
\]
\end{corollary}

\section{Name-independent
Input-Output Tasks Reduce to Leader Election}\label{app:ni-task}
As a consequence of the above, any \emph{name-independent} task can be solved in our model as long as leader election is possible. 
A task~$(\cI,\cO, \Delta)$ is name-independent if
$\Delta$ maps inputs to outputs in a name-oblivious way. Namely, for any possible input for the system, $\sigma\in\cI$, parties with the same input-value compute the same output-value, i.e., 
$(i,x),(j,x)\in \sigma \Rightarrow (i,o),(j,o) \in \Delta(\sigma)$.

by reducing the task to choosing a leader, who in turn computes the output in a centralized way.
\begin{theorem}
If Leader election is solvable by an anonymous network in the blackboard or message-passing model, then any distributed name-independent input-output task can be solved over the same model.
\end{theorem}
\begin{proof}
Given a task $(\cI,\cO,\Delta)$, the parties 
perform leader election. Every party then sends the leader its inputs, either directly or via the blackboard.
The leader collects all the inputs (and records the respective port-number for each in the message passing model). Then, the leader solves the task~$(\cI,\cO,\Delta)$ by himself and distributes the outputs to its neighbours, either by publishing the respective output of each input on the blackboard or by sending the appropriate output to the corresponding port. 
\end{proof}
It is obvious that the other direction is invalid, as there are tasks that can be deterministically solved in both these models, regardless of the solvability of leader election.

\end{document}

%% file: tikz_P.tex


\begin{tikzpicture}%
[pred/.style={circle,draw=red!80,fill=red!80,thick,
 pattern={custom north west lines},hatchcolor=red!80,	
inner sep=0pt,minimum size=2mm},
pblack/.style={circle,draw=black,fill=black!80,thick,
inner sep=0pt,minimum size=2mm},
pwhite/.style={circle,draw=black,fill=white,thick,
inner sep=0pt,minimum size=2mm},
alabel/.style={above,inner sep=5pt},
blabel/.style={below,inner sep=5pt}];

\draw (-3,3.5)  node  [pwhite] {} node [alabel] {$\bot$} --  (-1,3.5) node [pblack] {}  node [alabel] {$\bot$};

\draw (4,4)  node  [pwhite] {} node [alabel] {$\bot,0,(\bot)$} --  (6,4) node [pblack] {} node [alabel] {$\bot,0,(\bot)$} --
 (6,3) node [pwhite] {} node [blabel] {$\bot,1,(\bot)$} --  (4,3)  node  [pblack] {} node [blabel] {$\bot,1,(\bot)$} -- cycle;

\begin{scope}[yshift=-3cm]
	\draw (0,2)  node  [pwhite] {} node [alabel] {$k_0,0,(k_0)$} --  (2,2) node [pblack] {} node [alabel] {$k_0,0,(k_0)$} --  (2,1) node [pwhite] {} node [blabel] {$k_0,1,(k_0)$} --  (0,1)  node  [pblack] {} node [blabel] {$k_0,1,(k_0)$} -- cycle;
\end{scope}
\begin{scope}[yshift=-8cm]
	\draw (0,2)  node  [pwhite] {} node [alabel] {$k_1,0,(k_1)$} --  (2,2) node [pblack] {} node [alabel] {$k_1,0,(k_1)$} --  (2,1) node [pwhite] {} node [blabel] {$k_1,1,(k_1)$} --  (0,1)  node  [pblack] {} node [blabel] {$k_1,1,(k_1)$} -- cycle;
\end{scope}

\begin{scope}[yshift=-5.5cm, xshift=-4cm]
	\draw (0,2)  node  [pwhite] {} node [alabel] {$k_0,0,(k_1)$} --  (2,2) node [pblack] {} node [alabel] {$k_1,0,(k_0)$} --  (2,1) node [pwhite] {} node [blabel] {$k_0,1,(k_1)$} --  (0,1)  node  [pblack] {} node [blabel] {$k_1,1,(k_0)$} -- cycle;
\end{scope}
\begin{scope}[yshift=-5.5cm,xshift=4cm]
	\draw (0,2)  node  [pwhite] {} node [alabel] {$k_1,0,(k_0)$} --  (2,2) node [pblack] {} node [alabel] {$k_0,0,(k_1)$} --  (2,1) node [pwhite] {} node [blabel] {$k_1,1,(k_0)$} --  (0,1)  node  [pblack] {} node [blabel] {$k_0,1,(k_1)$} -- cycle;
\end{scope}

\draw [dashed] (-6,0.75) -- (8.75,0.75);
\draw [dashed] (1.5, 5.25) -- (1.5, 1.5);
\draw [densely dotted,-latex] (4.75,3) to [bend left=15]  (1.5,-5.1);
\draw [densely dotted,-latex] (6,3.5) to [out=0,in=70]  (5,-2.75);

\node at (-3.5, 5) {$\mathcal{P}(0)$};
\node at (7, 5) {$\mathcal{P}(1)$};
\node at (-3.5, -1) {$\mathcal{P}(2)$};

\end{tikzpicture}

%% file: tikz_R.tex


\usetikzlibrary{3d,perspective}


\begin{tikzpicture}%
[pred/.style={circle,draw=red!80,fill=red!80,thick,
 pattern={custom north west lines},hatchcolor=red!80,	
inner sep=0pt,minimum size=2mm},
pblack/.style={circle,draw=black,fill=black!80,thick,
inner sep=0pt,minimum size=2mm},
pwhite/.style={circle,draw=black,fill=white,thick,
inner sep=0pt,minimum size=2mm},
alabel/.style={above,inner sep=5pt},
blabel/.style={below,inner sep=5pt}];


\node at (-1,3)  {$\mathcal{R}(0)$};
\node at (5,3)  {$\mathcal{R}(1)$};

\coordinate (n1) at (0,0)  ;
\coordinate (n2) at (2,0)  ;
\coordinate (n3) at (60:2); 

\filldraw [draw=black,fill=black!50,fill opacity=0.25]   (n1)  -- (n2) -- (n3) -- cycle;
{\footnotesize
\node at (barycentric cs:n1=1,n2=1,n3=0.8)  {$(\bot,\bot,\bot)$};

\node at (n1)  [pwhite] {}  ;
\node at (n2) [pblack] {} ;
\node at (n3)  [pred] {} ; 
}

\begin{scope}[xshift=8cm,3d view]

\coordinate (B0) at (0,0,0) ;
\coordinate (W0) at (2,0,0);
\coordinate (R0) at (0,2,0);	
\coordinate (B1) at (2,2,2);
\coordinate (W1) at (0,2,2) ;
\coordinate (R1) at (2,0,2);

\filldraw [draw=black,fill=black!70,fill opacity=1,dotted] (B0) -- (R0) -- (W0) -- cycle;
\filldraw [draw=black,fill=black!70,fill opacity=1,dotted] (B1) -- (R0) -- (W0) -- cycle;

\filldraw [draw=black,fill=black!50,fill opacity=0.85,dotted] (B1) -- (R0) -- (W1) -- cycle;
\filldraw [draw=black,fill=black!50,fill opacity=0.85,dotted] (B1) -- (R1) -- (W0) -- cycle;

\filldraw [draw=black,fill=black!50,fill opacity=0.5] (B0) -- (R1) -- (W0) -- cycle;
\filldraw [draw=black,fill=black!50,fill opacity=0.5] (B0) -- (R0) -- (W1) -- cycle;

\filldraw [draw=black,fill=white!70,fill opacity=0.45] (B1) -- (R1) -- (W1) -- cycle;
\filldraw [draw=black,fill=white!70,fill opacity=0.45] (B0) -- (R1) -- (W1) -- cycle;

\node at (B0) [pblack,label=below:{(2,0)}] {} ;
\node at (W0) [pwhite,label=below:{(1,0)}] {};
\node at (R0) [pred,label=below:{(3,0)}] {};
\node at (B1) [pblack,label=above:{(2,1)}] {}  ;
\node at (W1) [pwhite,label=above:{(1,1)}] {}  ;
\node at (R1) [pred,label=above:{(3,1)}] {}  ;

{\footnotesize
\node at (barycentric cs:B0=1,W1=1,R1=1)  {$(1,0,1)$};
\node at (barycentric cs:B0=1,W0=1,R1=0.8)  {$(0,0,1)$};

\draw [latex-,dashed,bend left] (barycentric cs:B0=1,W1=0.5,R0=1) to (0,4,0.25) node [above] {$(1,0,0)$};
}

\end{scope}

\end{tikzpicture}

%% file: tikz_O_and_CO.tex

\usetikzlibrary {intersections,patterns}

\begin{tikzpicture}%
[pred/.style={circle,draw=red!80,fill=red!80,thick,
 pattern={custom north west lines},hatchcolor=red!80,	
inner sep=0pt,minimum size=2mm},
pblack/.style={circle,draw=black,fill=black!80,thick,
inner sep=0pt,minimum size=2mm},
pwhite/.style={circle,draw=black,fill=white,thick,
inner sep=0pt,minimum size=2mm}];


\coordinate (n10) at (0,0)  ;
\coordinate (n20) at (1,0)  ;
\coordinate (n30) at (60:1); 
\path (n20) ++ (60:1) coordinate (n11);
\coordinate (n21) at (120:1);
\coordinate (n31) at (-60:1) ;

\filldraw [draw=black,fill=black!50,fill opacity=0.25]   (n10)  -- (n21) -- (n30) -- cycle;
\filldraw [draw=black,fill=black!50,fill opacity=0.25]   (n20)  -- (n30) -- (n11) -- cycle;
\filldraw [draw=black,fill=black!50,fill opacity=0.25]   (n10)  -- (n20) -- (n31) -- cycle ;

\node at (n10)  [pwhite] {}  ;
\node at (n20) [pblack] {} ;
\node at (n30)  [pred] {} ; 
\node at (n11) [pwhite]  {};
\node at (n21) [pblack] {};
\node at (n31)  [pred] {};

{\footnotesize
\node at (barycentric cs:n10=1,n21=1,n30=1)  {$\tau_2$};
\node at (barycentric cs:n20=0.5,n11=0.5,n30=0.5)  {$\tau_1$};
\node at (barycentric cs:n10=0.5,n20=0.5,n31=0.5)  {$\tau_3$};
}

\node at (-1.5,-1.5) {$\mathcal{O}_{\mathsf{LE}}$};
{\footnotesize
\node at (n10)  [left,inner sep=5pt] {$(1,0)$}  ;
\node at (n20) [right,inner sep=5pt] {$(2,0)$} ;
\node at (n30)  [above,inner sep=5pt] {$(3,0)$} ; 
\node at (n11) [above,inner sep=5pt]  {$(1,1)$};
\node at (n21) [above,inner sep=5pt] {$(2,1)$};
\node at (n31)  [below,inner sep=5pt] {$(3,1)$};
}
\draw [rounded corners,ultra thick,black] (-2.3,-2) rectangle (2.5,2.5);
\end{tikzpicture}
\qquad\quad
\begin{tikzpicture}%
[pred/.style={circle,draw=red!80,fill=red!80,thick,
 pattern={custom north west lines},hatchcolor=red!80,	
inner sep=0pt,minimum size=2mm},
pblack/.style={circle,draw=black,fill=black!80,thick,
inner sep=0pt,minimum size=2mm},
pwhite/.style={circle,draw=black,fill=white,thick,
inner sep=0pt,minimum size=2mm}];


\coordinate (n10) at (0,0)  ;
\coordinate (n20) at (1,0)  ;
\coordinate (n30) at (60:1); 
\path (n20) ++ (60:1) coordinate (n11);
\coordinate (n21) at (120:1);
\coordinate (n31) at (-60:1) ;


\draw[black] (n30) -- (n10) -- (n20) -- cycle;

\node at (n10)  [pwhite] {}  ;
\node at (n20) [pblack] {} ;
\node at (n30)  [pred] {} ; 
\node at (n11) [pwhite]  {};
\node at (n21) [pblack] {};
\node at (n31)  [pred] {};


\node at (-1.5,-1.5) {$\pi(\mathcal{O}_{\mathsf{LE}})$};

{\footnotesize
\node at (n10)  [below left,inner sep=5pt] {$(1,0)$}  ;
\node at (n20) [below right,inner sep=5pt] {$(2,0)$} ;
\node at (n30)  [above,inner sep=10pt] {$(3,0)$} ; 
\node at (n11) [above right,inner sep=5pt]  {$(1,1)$};
\node at (n21) [above left,inner sep=5pt] {$(2,1)$};
\node at (n31)  [below,inner sep=8pt] {$(3,1)$};
}

\draw [rounded corners,ultra thick,black,use as bounding box] (-2.5,-2) rectangle (3,2.5);


\draw[name path=t1edge,thick,draw=red!60,dashed,opacity=0.5] (30:0.866) circle [x radius=0.8cm, y radius=3mm, rotate=-60]; 
\draw[name path=t1node,thick,draw=red!60,dashed,opacity=0.5] (n11) circle [radius=3mm];

\draw (-5.15,0.75) edge[bend left,very thick,draw=red!60,dashed,-stealth,shorten >=10pt] 
	node [pos=0.35,above,red!60]  {$\pi(\tau_1)$} 
	coordinate [pos=0.65] (cSplit)
	(n30);
\draw (cSplit) edge[bend left=45,very thick,draw=red!60,dashed,-stealth,shorten >=10pt,shorten <=5pt] 
	(n11);


\end{tikzpicture}